\documentclass[a4paper,11pt,reqno]{amsart}
\usepackage[utf8]{inputenc}
\usepackage{mathrsfs}
\usepackage{dsfont}
\usepackage{hyperref}
\usepackage{amsmath}
\usepackage{amssymb}
\usepackage{amsthm}
\usepackage{amsfonts}
\usepackage{amstext}
\usepackage{amsopn}
\usepackage{amsxtra}
\usepackage{mathrsfs}
\usepackage{dsfont}
\usepackage{esint}
\usepackage{enumitem}
\usepackage{graphicx}

\newtheorem{theorem}{Theorem}
\newtheorem{lemma}[theorem]{Lemma}

\newtheorem{remark}[theorem]{Remark}
\newtheorem{definition}[theorem]{Definition}

\newcommand{\R}{\mathbb{R}}
\newcommand{\Z}{\mathbb{Z}}
\newcommand{\C}{\mathbb{C}}
\newcommand{\N}{\mathbb{N}}
\newcommand{\bP}{\mathbb{P}}

\newcommand{\dps}{\displaystyle}
\newcommand{\ii}{\infty}
\newcommand\1{{\ensuremath {\mathds 1} }}
\renewcommand\phi{\varphi}

\newcommand{\wto}{\rightharpoonup}

\newcommand{\cQ}{\mathcal{Q}}

\newcommand{\cR}{\mathcal{R}}

\newcommand{\cK}{\mathcal{K}}
\newcommand{\cE}{\mathcal{E}}

\newcommand\pscal[1]{{\ensuremath{\left\langle #1 \right\rangle}}}
\newcommand{\norm}[1]{ \left\| #1 \right\|}
\newcommand{\tr}{{\rm Tr}\,}

\renewcommand{\geq}{\geqslant}
\renewcommand{\leq}{\leqslant}

\renewcommand{\tilde}{\widetilde}

\newcommand{\eps}{\varepsilon}
\usepackage{color}

\newcommand{\nn}{\nonumber}

\newcommand{\bx}{\mathbf{x}}
\newcommand{\bX}{\mathbf{X}}
\newcommand{\by}{\mathbf{y}}

\newcommand{\br}{\mathbf{r}}

\newcommand{\bk}{\mathbf{k}}
\newcommand{\bv}{\mathbf{v}}
\newcommand{\rd}{\mathrm{d}}

\title{Universal Functionals in Density Functional Theory}

\author[M. Lewin]{Mathieu Lewin}
\address{CNRS \& CEREMADE, Universit\'e Paris-Dauphine, PSL University, 75016 Paris, France}
\email{mathieu.lewin@math.cnrs.fr}

\author[E.H. Lieb]{Elliott H. Lieb}
\address{Departments of Mathematics and Physics, Jadwin Hall, Princeton University, Washington Rd., Princeton, NJ 08544, USA}
\email{lieb@princeton.edu}

\author[R. Seiringer]{Robert Seiringer}
\address{IST Austria (Institute of Science and Technology Austria), Am Campus 1, 3400 Klosterneuburg, Austria}
\email{robert.seiringer@ist.ac.at}

\date{\today}

\begin{document}

 \begin{abstract}
In this chapter we first review the Levy-Lieb functional, which gives the lowest kinetic and interaction energy that can be reached with all possible quantum states having a given density. We discuss two possible convex generalizations of this functional, corresponding to using mixed canonical and grand-canonical states, respectively. We present some recent works about the local density approximation, in which the functionals get replaced by purely local functionals constructed using the uniform electron gas energy per unit volume. We then review the known upper and lower bounds on the Levy-Lieb functionals. We start with the kinetic energy alone, then turn to the classical interaction alone, before we are able to put everything together. An appendix is devoted to the Hohenberg-Kohn theorem and the role of many-body unique continuation in its proof.

\bigskip

\noindent \sl \copyright~2019 by the authors. This paper may be reproduced, in its entirety, for non-commercial purposes.
 \end{abstract}

 \maketitle

 \tableofcontents

\section{Introduction}

Density Functional Theory (DFT) attempts to describe all the relevant information about a many-body quantum system at or near its ground state in terms of its one-body density $\rho(\br)$. In its orbital-free variational formulation by Levy~\cite{Levy-79} and Lieb~\cite{Lieb-83b}, DFT relies completely on a universal functional $\rho\mapsto F_{\rm LL}[\rho]$ which gives the lowest (kinetic plus interaction) energy that can be reached with all possible quantum states having a given density function $\br\mapsto\rho(\br)$. This functional is exact and it is able to describe interacting quantum Coulomb systems in their ground states. It is of course not known explicitly and one of the main purpose of DFT is to find suitable approximations. In this chapter we review known upper and lower bounds on this functional and discuss some regimes in which it simplifies. In particular we focus on the \emph{Local Density Approximation (LDA)} which becomes exact in the regime where the density is very flat on sufficiently large regions of space, as was recently proved in~\cite{LewLieSei-18,LewLieSei-19}. A simpler older approximation is the Thomas-Fermi functional~\cite{Thomas-27,Fermi-27} which is surprisingly accurate for heavy atoms~\cite{LieSim-73,LieSim-77b} and is reviewed in~\cite{Lieb-81b}.

It turns out that there are several possible Levy-Lieb-type functionals. Instead of considering $N$-particle wavefunctions one can as well work with mixed states~\cite{Lieb-83b}, or even with grand-canonical states for which only the average number of particles is fixed. The latter has not been thoroughly discussed in the literature. We consider the three possibilities in this chapter. When we discuss bounds, it is useful to consider the kinetic and interaction energies separately. This naturally brings in Lieb-Thirring~\cite{LieThi-75,LieThi-76} and Lieb-Oxford~\cite{Lieb-79,LieOxf-80} inequalities, which provide lower bounds on these two functionals.

In Appendix~\ref{sec:HK} we recall the Hohenberg-Kohn theorem, which is another important abstract result in DFT. It turns out that its proof relies on some unique continuation problems for $N$-particle systems, which are not yet completely understood.

\subsubsection*{Acknowledgments.} This project has received funding from the European Research Council (ERC) under the European Union's Horizon 2020 research and innovation programme (grant agreements MDFT No 725528 of M.L. and AQUAMS No 694227 of R.S.).

\subsubsection*{Notation.}  Everywhere in the chapter $\bx=(\br,\sigma)$ denotes both the space variable $\br\in\R^d$ and the spin variable $\sigma\in\Z_q$. Although the physical case corresponds to $d=3$ and $q=2$, it is sometimes useful to keep $d$ and $q$ general to better emphasize the role of the dimension and spin. To simplify our writing we use the convention that $\int_{\R^d\times\Z_q}f(\bx)\,\rd\bx=\sum_{\sigma\in\Z_q}\int_{\R^d}f(\br,\sigma)\,\rd\br$. For $N$ particles we use the notation $\bX=(\bx_1,...,\bx_N)\in(\R^d\times\Z_q)^N$ with a similar convention for the integral $\int_{(\R^d\times\Z_q)^N}$.

We recall that the one-body density $\rho_\Psi$ of an $N$-particle fermionic (normalized) wavefunction
$$\Psi\in \bigwedge_1^NL^2(\R^d\times\Z_q,\C)$$
is defined by
$$\rho_\Psi(\br):=N\sum_{\sigma_1,...,\sigma_N\in\Z_q}\int_{\R^d}\cdots \int_{\R^d}|\Psi(\br,\sigma_1,\br_2,\sigma_2,\dots,\br_N,\sigma_N)|^2\,\rd\br_2\cdots \rd\br_N$$
The interpretation of $\rho_\Psi$ is that it provides the average number of particles in space, without taking their spin into account.

\section{Universal functionals in Density Functional Theory}

Following~\cite{Levy-79,Lieb-83b} we express the ground state energy as a variational problem involving only the density, and we discuss some (known and unknown) mathematical properties of the associated  functionals.

\subsection{The Levy-Lieb universal functional}
For completeness we consider general interaction potentials $w$ in any space dimension $d$. This is useful to understand the particularities of the physical case at work in DFT, namely the Coulomb interaction $w(\br)=|\br|^{-1}$ in dimension $d=3$.

Let
$$v_+\in L^1_{\rm loc}(\R^d,\R),\qquad v_-,w\in L^p(\R^d,\R)+L^\ii(\R^d),$$
with $w$ even, with $v_\pm \geq 0$, with $v := v_{+} - v_{-}$ and with
\begin{equation}
p\begin{cases}
=1&\text{when $d=1$,}\\
>1&\text{when $d=2$,}\\
=\frac{d}{2}&\text{when $d\geq3$.}
    \end{cases}
    \label{eq:hyp_p}
\end{equation}
Under the assumption~\eqref{eq:hyp_p}, the $N$-body potential
$$W_N^{v,w}(\br_1,...,\br_N):=\sum_{j=1}^Nv(\br_j)+\sum_{1\leq j< k\leq N}w(\br_j-\br_k)$$
is infinitesimally $(-\Delta)$--form bounded from below, which means that
$$\int_{(\R^d\times\Z_q)^N}W_N^{v,w}|\Psi|^2\geq -\eps\int_{(\R^d\times\Z_q)^N}|\nabla\Psi|^2-C_{N,\eps}\int_{(\R^d\times\Z_q)^N}|\Psi|^2$$
for all $\Psi$ and all $\eps>0$, with $C_{N,\eps}$ a constant depending only on $\eps$ and $N$. From this we can deduce that the quadratic form associated with the symmetric $N$-body operator
\begin{equation}\label{eq:H_Nvw}
H^{v,w}_N:=\sum_{j=1}^N\left(-\frac{\Delta_{\br_j}}2+v(\br_j)\right)+\sum_{1\leq j< k\leq N}w(\br_j-\br_k)
\end{equation}
is closed on the energy space
\begin{multline}
\cQ(H_N^{v,w}):=\bigg\{\Psi\in \bigwedge_1^NL^2(\R^d\times\Z_q)\ :\\ \int_{(\R^d\times\Z_q)^N}|\nabla\Psi(\bX)|^2\,\rd\bX+\int_{\R^d}v_+(\br)\,\rho_\Psi(\br)\,\rd\br<\ii\bigg\}.
\label{eq:form_domain}
\end{multline}
This allows us to work with the associated \emph{Friedrichs self-adjoint realization} of $H_N^{v,w}$, see~\cite[Sec.~VIII.6]{ReeSim1} and~\cite[Sec.~X.3]{ReeSim2}.
Everywhere we work with fermions, that is, the wavefunction $\Psi$ is assumed to be anti-symmetric with respect to exchanges of its variables. Note that the bosonic case is obtained when the number of spin states $q$ is equal to $N$, see~\cite[Sec.~3.1.2--3.1.3]{LieSei-09}.

The Hamiltonian~\eqref{eq:H_Nvw} describes $N$ fermionic particles evolving in $\R^d$, with $q$ spin states, submitted to an external potential $v$ and interacting via the pair potential $w$. An important physical quantity is the \emph{ground state energy} of the system which is obtained by minimizing the corresponding energy
$$\boxed{E_N[v]:=\inf_{\substack{\Psi\in \cQ(H^{v,w}_N)\\ \int_{(\R^d\times\Z_q)^N}|\Psi|^2=1}}\pscal{\Psi,H^{v,w}_N\Psi}}$$
where $\pscal{\Psi,H^{v,w}_N\Psi}$ is understood in the sense of quadratic forms. By the variational principle, this is just the bottom of the spectrum of the operator~$H^{v,w}_N$:
$$E_N[v]:=\min\sigma\left(H^{v,w}_N\right)$$
in the fermionic subspace. We do not emphasize the interaction $w$ in our notation since it will usually be fixed. When needed we will instead use the notation $E_N^w[v]$.

The main idea of~\cite{Levy-79,Lieb-83b} is to replace the infimum over $\Psi$ by a two-step minimization
$$\inf_\Psi\;(\cdots)= \inf_{\rho}\inf_{\substack{\Psi\\ \rho_\Psi=\rho}}(\cdots)$$
where we first minimize over the density $\rho$ and then over all the wavefunctions having this prescribed density. This procedure requires to identify the set of \emph{$N$-representable densities}, that is, those arising from a $\Psi$ in the form domain of $H_N^{v,w}$, a question that we now address.

The Hoffmann-Ostenhof inequality~\cite{Hof-77} states that
\begin{equation}
 \sum_{j=1}^N\int_{(\R^d\times\Z_q)^N}|\nabla_j \Psi(\bX)|^2\,\rd\bX\geq \int_{\R^d}\left|\nabla\sqrt{\rho_\Psi}(\br)\right|^2\,\rd\br
 \label{eq:Hoffmann-Ostenhof}
\end{equation}
for all (bosonic or fermionic) wavefunctions (see also Theorem~\ref{thm:HO} below). This inequality implies that we should restrict ourselves to densities such that
$$\sqrt{\rho}\in H^1(\R^d)=\left\{f\in L^2(\R^d)\ :\ \nabla f\in L^2(\R^d)\right\}.$$
This turns out to be the optimal condition.

\begin{theorem}[Representability of the one-particle density~\cite{Lieb-83b}]\label{thm:representability}
Let $\rho\in L^1(\R^d,\R_+)$ be such that $\sqrt\rho\in H^1(\R^d)$ and $\int_{\R^d}\rho(\br)\,\rd\br=N\in\mathbb{N}$. Then there exists one normalized antisymmetric wavefunction $\Psi\in \bigwedge_1^NL^2(\R^d\times\Z_q,\C)$ of finite kinetic energy, $\int_{(\R^d\times\Z_q)^N}|\nabla\Psi|^2<\ii$, such that $\rho=\rho_\Psi$.
\end{theorem}

The proof of the theorem is much easier for $q\geq N$ where the antisymmetry can be put entirely in the spin variables. One can just take
$$\Psi(\bX)=\prod_{j=1}^N\sqrt{\frac{\rho(\br_j)}{N}}\;\frac{\det(\delta_j(\sigma_k))_{1\leq j,k\leq N}}{\sqrt{N!}}.$$
When $q<N$ the proof in~\cite{Lieb-83b}, inspired of~\cite{MarYou-58,Harriman-81}, consists of taking a Slater determinant
$$\Psi(\bX)=\frac1{\sqrt{N!}}\det\{\phi_j(\bx_k)\}$$
with the orbitals
\begin{equation}
 \phi_j(\bx)=\sqrt{\frac{\rho(\br)}{N}}e^{i\theta_j(\br)}\delta_0(\sigma)
 \label{eq:phases_pre}
\end{equation}
where the phases $\theta_j$ are chosen to make the $\phi_j$ orthonormal. Although one knows that there exist such phases (an explicit example will be given later in~\eqref{eq:phases}), one has a very bad control of their behavior in $N$. This will be discussed later in Section~\ref{sec:kinetic} when we consider the kinetic energy cost of introducing such phases.

At this point we have found the set of $N$-representable densities
\begin{equation}
 \left\{\rho\in L^1(\R^d,\R_+)\ :\ \int_{\R^d}\rho=N,\ \int_{\R^d}|\nabla\sqrt\rho|^2<\ii\right\}.
 \label{eq:N-representable}
\end{equation}
Note that this is a convex set since $\rho\mapsto\int_{\R^d}|\nabla\sqrt\rho|^2$ is convex by~\cite[Thm.~7.8]{LieLos-01}. With Theorem~\ref{thm:representability} at hand, we can rewrite the ground state energy as a minimization principle over $\rho$ in this convex set:
\begin{equation}
\boxed{E_N[v]:=\inf_{\substack{\sqrt\rho\in H^1(\R^d)\\ \int_{\R^d}\rho=N\\ \int_{\R^d}v_+\rho<\ii}}\left\{F_{\rm LL}[\rho]+\int_{\R^d}v(\br)\,\rho(\br)\,\rd\br\right\} }
\label{eq:relation_E_F_LL}
\end{equation}
where
\begin{multline}
F_{\rm LL}[\rho]:=\inf_{\substack{\Psi\in \bigwedge_1^NL^2(\R^d\times\Z_q,\C)\\ \norm{\Psi}_{L^2}=1\\ \rho_\Psi=\rho}}\Bigg\{\frac12\sum_{j=1}^N\int_{(\R^d\times\Z_q)^N}|\nabla_j\Psi(\bX)|^2\,\rd\bX\\
+\sum_{1\leq j<k\leq N}\int_{(\R^d\times\Z_q)^N} |\Psi(\bX)|^2w(\br_j-\br_k)\,\rd\bX\Bigg\}
\label{eq:F_LL}
\end{multline}
is called the \emph{Levy-Lieb functional}. It is the lowest  possible (kinetic plus interaction) energy of a quantum system having the prescribed density $\rho$. This universal functional is the central object of DFT, since knowing it would allow one to compute the ground state energy of a system with any external potential $v$, by~\eqref{eq:relation_E_F_LL}. In this chapter we will review what is known about $F_{\rm LL}$. But first we need to introduce two other universal functionals, which are obtained after convexifying $F_{\rm LL}$ in two different ways.

\subsection{Lieb's universal functional}
Note that $v\mapsto E_N[v]$ is concave (as seen from the variational principle, it is a minimization over $\Psi$ of affine functions in $v$). More precisely, from~\eqref{eq:relation_E_F_LL} we see that $E_N$ is the Legendre transform of $F_{\rm LL}$ on the convex set of $N$-representable densities. This naturally raises the question of whether $F_{\rm LL}$ is, conversely, the Legendre transform of $E_N$. This turns out to be \emph{wrong} since $F_{\rm LL}$ is \emph{not convex}~\cite{Lieb-83b}. It is therefore convenient to look at the convex hull (also called lower convex envelope)
$$F_{\rm L}:={\rm Conv}(F_{\rm LL}),$$
which is the Legendre transform of $E_N$.
Here the convex hull means that it is the largest convex function below $F_{\rm LL}$. We always assume that $\int_{\R^d}\rho=N$. Another kind of convex hull with respect to $N$ will be considered later.
As proved in~\cite{Lieb-83b}, the function $F_{\rm L}$ is explicit and given by a similar definition as in~\eqref{eq:F_LL} but with mixed states instead of pure states:
\begin{equation}
\boxed{F_{\rm L}[\rho]:=\inf_{\substack{\Gamma=\Gamma^*\geq0\\ \tr\Gamma=1\\ \tr(-\Delta)\Gamma<\ii\\ \rho_\Gamma=\rho}}\tr\left(H^{0,w}_N\Gamma\right).}
\label{eq:F_L}
\end{equation}
We recall that the density $\rho_\Gamma$ of a mixed state $\Gamma$ (a non-negative self-adjoint operator satisfying $\tr(\Gamma)=1$), diagonalized in the form
$$\Gamma=\sum_j \alpha_j|\Psi_j\rangle\langle \Psi_j|$$
with $\alpha_j\geq0$ and $\sum_j \alpha_j=1$, is defined by
$$\rho_ \Gamma:=\sum_j \alpha_j\, \rho_{\Psi_j}.$$

It is useful to know that the infimum is attained in~\eqref{eq:F_L}, as well as for the Levy-Lieb functional in~\eqref{eq:F_LL}.

\begin{theorem}[Existence of optimal (pure and mixed) states~\cite{Lieb-83b}]\label{thm:exists_min}
Let $\rho\in L^1(\R^d,\R_+)$ be such that $\int_{\R^d}\rho=N\in\N$ and $\sqrt\rho\in H^1(\R^d)$. Then the infima in~\eqref{eq:F_LL} and~\eqref{eq:F_L} are attained.
\end{theorem}

The proof uses the fact that a minimizing sequence $\Psi_j$ (resp.~$\Gamma_j$) is necessarily compact in $L^2(\R^d\times\Z_q,\C)^N$ (resp. in the trace class), since the density is fixed, hence the sequence is tight.

Let us consider a density $\rho$ and a corresponding minimizing $N$-particle mixed state $\Gamma$. Diagonalizing $\Gamma$ in the form $\Gamma=\sum_{j}\alpha_j|\Psi_j\rangle\langle \Psi_j|$, we see that
$$F_{\rm L}[\rho]=\sum_{j}\alpha_j\pscal{\Psi_j,H_N^{0,w}\Psi_j}\geq \sum_{j}\alpha_j\,F_{\rm LL}[\rho_{\Psi_j}]$$
and since the upper bound is obvious, we conclude that~\eqref{eq:F_L} can also be written in the form
\begin{equation}
 F_{\rm L}[\rho]=\min_{\substack{\rho=\sum_j\alpha_j\,\rho_j\\\sum_{j}\alpha_j=1\\  \sqrt{\rho_j}\in H^1(\R^d)\\ \int_{\R^d}\rho_j=N}}\sum_j\alpha_j\,F_{\rm LL}[\rho_j].
 \label{eq:convex_hull_L}
\end{equation}
This is the claimed convex hull  of $F_{\rm LL}$.

Since an affine function always attains its minimum at an extreme point of a convex set, the ground state energy $E_N[v]$ is given by the same formula
\begin{align*}
E_N[v]&=\inf_{\substack{\Gamma=\Gamma^*\geq0\\ \tr\Gamma=1\\ \tr(-\Delta)\Gamma<\ii\\ \rho_\Gamma=\rho}}\tr\left(H^{v,w}_N\Gamma\right)\\
&=\inf_{\substack{\sqrt\rho\in H^1(\R^d)\\ \int_{\R^d}\rho=N\\ \int_{\R^d}v_+\rho<\ii}}\left\{F_{\rm L}[\rho]+\int_{\R^d}v(\br)\,\rho(\br)\,\rd\br\right\}
\end{align*}
as we had in~\eqref{eq:relation_E_F_LL} for pure states. From this discussion it seems more natural to work with the convex Lieb functional $F_{\rm L}$, instead of $F_{\rm LL}$. The following duality principle holds.

\begin{theorem}[Duality~\cite{Lieb-83b}]\label{thm:dual}
We have
\begin{align}
 F_{\rm L}[\rho]&=\sup_{v\in L^p(\R^d)+L^\ii(\R^d)}\left\{E_N[v]-\int_{\R^d}v(\br)\rho(\br)\,\rd\br\right\}\nn\\
 &=\sup_{\substack{v\in L^p(\R^d)+L^\ii(\R^d)\\ H^{v,w}_N\geq0}} \left\{-\int_{\R^d}v(\br)\rho(\br)\,\rd\br\right\}\label{eq:duality}
\end{align}
with $p$ as in~\eqref{eq:hyp_p}.
\end{theorem}

In the second line the constant $E_N[v]$ has been included in $v$, hence the constraint that $H^{v,w}_N\geq0$ in the operator sense.

We have seen in Theorem~\ref{thm:exists_min} that~\eqref{eq:F_LL} and~\eqref{eq:F_L} are attained. On the other hand, the supremum in~\eqref{eq:duality} will \emph{not} be attained for most densities. Indeed, $\rho$ would then be the density of a mixed ground state for the corresponding $E_N[v]$ but this set is believed to be very small. For instance, if $\rho$ vanishes on a set of positive measure, the supremum cannot be attained with a $v$ for which unique continuation holds on the whole space. This is discussed in Appendix~\ref{sec:HK}.

The importance in applications of the convex formulation of DFT based on the functionals $F_{\rm LL}$ and $F_{\rm L}$ is reviewed in~\cite{HelTea-22}.

\subsection{Grand canonical universal functional}

At this step we have defined two universal functionals $F_{\rm LL}$ and $F_{\rm L}$ of the density $\rho$, which provide the same ground state energy $E_N[v]$ in presence of an external potential $v$. Since Lieb's functional $F_{\rm L}$ is convex, it is to be preferred over $F_{\rm LL}$. The convexity implies the dual formula stated in Theorem~\ref{thm:dual}.

In spite of its convexity, the functional $F_{\rm L}$ does not behave well with respect to the weak topology of $H^1(\R^d)$. This might lead to some difficulties in processes where some electrons are lost, e.g. for scattering. Indeed, if we have a sequence $\rho_n$ such that $\sqrt{\rho_n}\wto\sqrt{\rho}$ weakly but not strongly, then we may have $\int_{\R^d}\rho(\br)\,\rd\br<N$. The integral does not even need to be an integer, in which case $F_{\rm L}[\rho]$ would not make any sense. For this reason it is natural to introduce a functional allowing for non-integer values of the particle number. The following grand-canonical version did not appear in~\cite{Lieb-83b}, but  was mentioned in~\cite{PerParLevBal-82} and was recently studied in~\cite{LewLieSei-18,LewLieSei-19}.

For us, a grand-canonical state (commuting with the particle number) will be a collection $\Gamma=(\Gamma_n)_{n\geq0}$ of non-negative self-adjoint operators, each of them acting on the $n$-particle space $\bigwedge_1^nL^2(\R^d\times\Z_q)$, and such that
$$\Gamma_0+\sum_{n\geq1}\tr(\Gamma_n)=1.$$
Here $\Gamma_0$ is a number in $[0,1]$ which gives the probability that there is no particle at all.
The corresponding density is the sum
$$\rho_\Gamma=\sum_{n\geq1}\rho_{\Gamma_n}$$
so that the average number of particles in the system is given by $\int_{\R^d}\rho_\Gamma=\sum_{n\geq1}n\tr(\Gamma_n)$. We define the grand-canonical universal functional as
\begin{equation}
\boxed{F_{\rm GC}[\rho]:=\inf_{\substack{\sum_{n\geq1}\tr\Gamma_n\leq1\\ \sum_{n\geq1}\tr(-\Delta)\Gamma_n<\ii\\ \sum_{n\geq1}\rho_{\Gamma_n}=\rho}}\sum_{n\geq1}\tr\left(H^{0,w}_n\Gamma_n\right).}
\label{eq:F_GC}
\end{equation}
In order to guarantee that the problem is well posed, we need some more assumptions on the interaction potential $w$. We assume that the system is \emph{stable of the second kind}~\cite{Ruelle}, that is, there exists a constant $C$ such that
\begin{equation}
\forall n\geq1,\qquad H_n^{0,w}\geq -Cn.
\label{eq:stability}
\end{equation}
Inserting this in~\eqref{eq:F_GC}  implies
$$F_{\rm GC}[\rho]\geq -C\int_{\R^d}\rho(\br)\,\rd\br$$
hence the infimum in~\eqref{eq:F_GC} is finite. A typical example is that of a non-negative interaction potential $w\geq0$ such as Coulomb, or more generally a potential $w$ which is classically stable of the second-kind~\cite{Ruelle}, that is, which satisfies the same assumption as~\eqref{eq:stability} with the kinetic energy removed:
\begin{equation}
\forall n\geq 2,\qquad \sum_{1\leq j<k\leq n}w(\br_j-\br_k)\geq -Cn \qquad\text{a.e. on $(\R^d)^n$.}
\label{eq:w_class_stable}
\end{equation}

For a density with integer particle number $\int_{\R^d}\rho(\br)\,\rd\br=N\in\mathbb{N}$ the grand canonical functional is the lowest of the three universal functionals:
$$F_{\rm GC}[\rho]\leq F_{\rm L}[\rho]\leq F_{\rm LL}[\rho].$$
The following can be shown similarly as for Theorem~\ref{thm:exists_min}, using the stability assumption~\eqref{eq:stability}.

\begin{theorem}[Existence of optimal grand-canonical states]\label{thm:exists_min_GC}
Assume that the system is stable of the second kind as in~\eqref{eq:stability}. Let $\rho\in L^1(\R^d,\R_+)$ be such that $\sqrt\rho\in H^1(\R^d)$. Then the infimum in~\eqref{eq:F_GC} is attained.
\end{theorem}


From the existence of a minimizer we deduce as before that
$$F_{\rm GC}[\rho]=\min_{\substack{\rho=\sum_n\alpha_n\,\rho_n\\ \sum_{n}\alpha_n=1\\  \sqrt{\rho_n}\in H^1(\R^d)\\ \int_{\R^d}\rho_n=n}}\sum_n\alpha_n\,F_{\rm L}[\rho_n]=\min_{\substack{\rho=\sum_j\beta_j\,\rho_j\\ \sum_{j}\beta_j=1\\  \sqrt{\rho_j}\in H^1(\R^d)\\ \int_{\R^d}\rho_j\in\N}}\sum_j\beta_j\,F_{\rm LL}[\rho_j].$$
In other words, the grand canonical functional is also a convex hull of the original Levy-Lieb functional $F_{\rm LL}$, but convex combinations $\rho=\sum_j\beta_j\,\rho_j$ are considered with the $\rho_j$ having an arbitrary number of particles. It is not required that all the $\rho_j$ have the fixed number $N$ of particles like for $F_{\rm L}$.

An interesting question is to determine whether an optimal grand canonical state $(\Gamma_n)_{n\geq0}$ corresponding to a given $\rho$ always satisfies $\Gamma_n\equiv0$ for $n\geq n_{\rm max}$. In this case we say that $\Gamma=(\Gamma_n)_{n\geq0}$ has a compact support in $n$.
No result of this sort seems to have appeared in the literature up to now.

The following theorem asserts that $F_{\rm GC}$ is the weak-$\ast$ lower semi-continuous envelope of $F_{\rm L}$, in an appropriate sense.

\begin{theorem}[Weak lower semi-continuity]\label{thm:wlsc_F_GC}
We assume that
$$0\leq w\in L^p(\R^d)+L^\ii(\R^d)$$
with $p$ as in~\eqref{eq:hyp_p} and that $w$ tends to 0 at infinity. The functional $F_{\rm GC}$ is the weak-$\ast$ lower semi-continuous closure of $F_{\rm L}$, in the following sense:

\smallskip

\noindent $(i)$ For any sequence $(\sqrt{\rho_j})_{j\geq1}\subset H^1(\R^d)$ converging weakly in $\dot{H}^1(\R^d)$ to $\sqrt\rho\in H^1(\R^d)$, we have
\begin{equation}
 F_{\rm GC}[\rho]\leq \liminf_{j\to\ii} F_{\rm GC}[\rho_j].
 \label{eq:wlsc}
\end{equation}

\smallskip

\noindent $(ii)$ For any $\sqrt\rho\in H^1(\R^d)$, there exists a sequence $(\sqrt{\rho_j})_{j\geq1}\subset H^1(\R^d)$ converging strongly to $\sqrt\rho$ in $\dot{H}^1(\R^d)\cap L^{p}(\R^d)$ for all $2<p<p^*$, such that
\begin{equation}
 F_{\rm GC}[\rho]=\lim_{j\to\ii}F_{\rm L}[\rho_j].
 \label{eq:approx_wlsc}
\end{equation}
\end{theorem}

Here
$$p^*=\begin{cases}
\ii&\text{in dimensions $d=1,2$,}\\
\frac{2d}{d-2}&\text{in dimensions $d\geq3$,}
\end{cases}$$
is the critical Sobolev exponent. Since we are not aware that the proof of Theorem~\ref{thm:wlsc_F_GC} was explicitly written anywhere, we provide the full argument later in Appendix~\ref{app:wlsc_F_GC}. It is inspired by~\cite{Lewin-11}.

When the liminf on the right of~\eqref{eq:wlsc} is finite (which we can always assume, otherwise the statement is void), then the Hoffmann-Ostenhof inequality~\eqref{eq:Hoffmann-Ostenhof} implies that $\sqrt{\rho_j}$ is bounded in the homogeneous Sobolev space $\dot{H}^1(\R^d)$. However $\int_{\R^d}\rho_j$ need not be bounded in general and this is why only the weak convergence in $\dot{H}^1(\R^d)$ was assumed. This plays an important role in~$(ii)$. Consider a $\rho$ and an associated optimal grand-canonical state $\Gamma=(\Gamma_n)_{n\geq0}$ for $F_{\rm GC}[\rho]$. If we have $\Gamma_{n_k}\neq0$ for a sequence $n_k\to\ii$, then this means that infinitely many particles are needed to properly represent $\rho$ grand-canonically. Although this is very unlikely to happen in practical situations, this can probably not be avoided for a general interaction $w$ and a general $\rho$. Then we need a diverging number of particles
$$\int_{\R^d}\rho_j=N_j\to+\ii$$
in our canonical state associated with $\rho_j$, even if it has a bounded energy. On the other hand, if there exists one minimizer $(\Gamma_n)_{n\geq0}$ for $F_{\rm GC}[\rho]$ which has a compact support in $n$, then $(ii)$ holds with a sequence converging weakly in $H^1(\R^d)$, as will be clear from our proof. This is one reason why it is important to understand whether optimal states always have a compact support in $n$, as  we have mentioned previously.

Next we discuss the dual formulation of $F_{\rm GC}$. For $v=v_+-v_-$ with $v_-\in L^p(\R^d)+L^\ii(\R^d)$ and $v_+\in L^1_{\rm loc}(\R^d)$, we find that the Legendre transform of $F_{\rm GC}$ is given by
\begin{equation}
E^{\rm GC}_\lambda[v]:=\inf_{\substack{\sqrt\rho\in H^1(\R^d)\\ \int_{\R^d}\rho=\lambda\\ \int_{\R^d}v_+\rho<\ii}}\left\{F_{\rm GC}[\rho]+\int_{\R^d}v\rho\right\}=\inf_{\substack{\sum_n\alpha_n=1\\ \sum_n n\alpha_n=\lambda}}\sum_n\alpha_nE_n[v].
\label{eq:relation_E_F_GC}
\end{equation}
For $\lambda=N\in\N$, we have $E^{\rm GC}_N[v]\leq E_N[v]$ but equality will in general not hold. If the function $n\mapsto E_n[v]$ is convex in the discrete sense, that is,
\begin{equation}
 E_n[v]-E_{n-1}[v]\leq E_{n+1}[v]-E_n[v],\qquad \forall n\geq1,
 \label{eq:conjecture_convexity_E_n}
\end{equation}
then it follows that
$$E^{\rm GC}_{N+\theta}[v]=(1-\theta)E_N[v]+\theta E_{N+1}[v]$$
for all $N\in\N$ and $\theta\in[0,1)$~\cite{PerParLevBal-82}. In particular $E_N^{\rm GC}[v]=E_N[v]$.
Therefore,  in case~\eqref{eq:conjecture_convexity_E_n} holds, the grand canonical functional $F_{\rm GC}[\rho]$ provides the same ground state energy in an external potential $v$ as the canonical ones $F_{\rm LL}[\rho]$ and $F_{\rm L}[\rho]$.
In physical terms the condition~\eqref{eq:conjecture_convexity_E_n} means that the electron ionization energy is greater than or equal to the electron affinity. It is a famous conjecture that~\eqref{eq:conjecture_convexity_E_n} holds for the Coulomb potential in dimension $d=3$, for atomic or molecular external potentials $v$~\cite{PerParLevBal-82,BacDel-14}. A counterexample is provided in~\cite{Lieb-83b} for a different $w$. Note that~\eqref{eq:conjecture_convexity_E_n} always holds for $w\equiv0$.

In case~\eqref{eq:conjecture_convexity_E_n} does \emph{not} hold, then $E^{\rm GC}_\lambda[v]$ is equal to the convex hull of $n\mapsto E_n[v]$. This amounts to considering the set $\cK\subset\N$ of the points $k$ such that
$$E_k[v]-E_{k-1}[v]=\max_{n=2,...,k} \left( E_n[v]-E_{n-1}[v]\right) .$$
One obtains
$$E^{\rm GC}_N[v]= E_{n_1}[v]+\frac{E_{n_1}[v]-E_{n_2}[v]}{n_1-n_2} (N-n_1)$$
where $n_1<N<n_2$ are the two closest points in $\cK$ on the left and right of $N$.

\subsection{Kohn-Sham exchange correlation}\label{sec:KS}

In the previous sections we have explained the Levy-Lieb variational formulation of the ground state energy of the $N$-particle problem in terms of the density only, which is really in the spirit of DFT. Practitioners prefer to use an auxiliary  set of $N$ orthonormal functions $\Phi=(\phi_1,...,\phi_N)$, which describe $N$ fictitious uncorrelated electrons, to build the desired density through the formula
$$\rho_\Phi(\br)=\sum_{n=1}^N\sum_{\sigma\in\Z_q}|\phi_n(\br,\sigma)|^2=\rho_\Psi(\br)$$
with the \emph{Slater determinant}
\begin{equation}
\Psi(\bx_1,\dots,\bx_N)=\frac{\det(\phi_j(\bx_k))}{\sqrt{N!}}.
 \label{eq:Slater_det}
\end{equation}
This method provides a better representation of the kinetic energy, but it is much more costly from a computational point of view. We quickly explain this approach due to Kohn-Sham~\cite{KohSha-65} here.

For a density $\rho$ with $\int_{\R^d}\rho(\br)d\br=N\in\N$, we introduce the lowest kinetic energy of Slater determinants
\begin{equation}
T_{\rm S}[\rho]:=\min_{\substack{\phi_1,\dots,\phi_N\in H^1(\R^d\times\Z_q,\C)\\ \pscal{\phi_i,\phi_j}=\delta_{ij}\\\rho_\Phi=\rho}}\frac12\sum_{j=1}^N\int_{\R^d\times\Z_q}|\nabla\phi_j(\bx)|^2\,\rd\bx
\label{def:Slater_kinetic_energy}
\end{equation}
(the min is attained for the same reason as in Theorem~\ref{thm:exists_min}). We then add and subtract $T_{\rm S}$ from $F_{\rm LL}$, which allows rewriting the $N$-particle ground state using $N$ orbitals as
\begin{multline}
E_N[v]:=\inf_{\substack{\phi_1,\dots,\phi_N\in H^1(\R^d\times\Z_q,\C)\\ \pscal{\phi_i,\phi_j}=\delta_{ij}\\ \int_{\R^d}\rho_\Phi\,v_+<\ii}}\bigg\{\frac12\sum_{j=1}^N\int_{\R^d\times\Z_q}|\nabla\phi_j(\bx)|^2\,\rd\bx+\int_{\R^d}v(\br)\rho_\Phi(\br)\,\rd\br\\
+\frac12\int_{\R^d}\int_{\R^d}w(\br-\br')\rho_\Phi(\br)\rho_\Phi(\br')\,\rd\br\,\rd\br'+E_{\rm xc}[\rho_\Phi]\bigg\},
\end{multline}
where
\begin{equation}
 E_{\rm xc}[\rho]:=F_{\rm LL}[\rho]-T_{\rm S}[\rho]-\frac12\int_{\R^d}\int_{\R^d}w(\br-\br')\rho(\br)\rho(\br')\,\rd\br\,\rd\br'
 \label{def:xc}
\end{equation}
is called the \emph{exchange-correlation energy}. From a mathematical point of view, the Kohn-Sham approach \emph{a priori} requires to study both $F_{\rm LL}[\rho]$ and $T_{\rm S}[\rho]$ as separated functionals. It is an interesting question to find a way to study $E_{\rm xc}[\rho]$ directly, without interpreting it as a difference.
In chemistry one often relies on the \emph{adiabatic connection formula} (see Remark~\ref{rmk:adiabatic_conn} below), which however involves another kinetic energy functional $T[\rho]$ discussed later in Section~\ref{sec:kinetic}.

Instead of using $N$ uncorrelated electrons as main variable, one may also use a one-particle density matrix $\gamma$ as main variable, which is often called the \emph{Kohn-Sham method with fractional occupations}. The method is similar but one has to subtract the lowest kinetic energy $T_{\rm GC}[\rho]$ of all possible one-particle density matrices, which is defined later in Section~\ref{sec:kinetic}.

\section{The Uniform Electron Gas and the Local Density Approximation}\label{sec:LDA}

The universal functionals $F_{\rm LL}$, $F_{\rm L}$ and $F_{\rm GC}$ defined in the previous section allow in principle to describe any fermionic system interacting via the potential $w$. But these functionals are of course not known exactly and finding them is essentially the same as solving the $N$-particle problem. One of the main purpose of DFT is to find reliable and efficient approximations. Here we discuss the most widely used of these approximations, called the \emph{Local Density Approximation (LDA)}~\cite{HohKoh-64,KohSha-65,LunMar-83,DreGro-90,ParYan-94,PerKur-03}, where they are replaced by purely local ones. The LDA is often considered as \emph{``the mother of all approximations''}~\cite{PerSch-01} and it yields surprisingly good results, even in cases where the density is not at all slowly varying~\cite{LunMar-83,ParYan-94}. Its successors involving gradient corrections are even better and have become the standard in DFT calculations. In this section we only consider the Coulomb case in dimension $d=3$ but we expect similar results for other potentials in all dimensions.

Of course, the functionals $F_{\rm LL}$, $F_{\rm L}$ and $F_{\rm GC}$ are not local at all. Two electrons at different places are always entangled and, furthermore, the Coulomb potential has a very long range so that electrons interact even when they are far apart. In the LDA one makes the assumption that the only non-local part is the Hartree term (the classical Coulomb energy of the density $\rho$) and one approximates the rest by a local function of $\rho$, that is, the integral of a function $f$ depending only on the value $\rho(\br)$ at $\br$:
\begin{equation}
F_\text{LL,L,GC}[\rho]\ \ \approx\ \underbrace{\frac12\iint_{\R^3\times \R^3}\frac{\rho(\br)\,\rho(\br')}{|\br-\br'|}\,\rd\br\,\rd\br'}_{\substack{\text{\bf non local}\\ \text{classical Coulomb energy of $\rho$}}} \ \ +\underbrace{\int_{\R^3}f\big(\rho(\br)\big)\,\rd\br.}_{\substack{\text{\bf local}\\ \text{$f$ = energy per unit vol.}\\ \text{of uniform electron gas}}}
\end{equation}
The function $f$ is chosen to be the energy per unit volume of an infinite gas of constant density $\rho$, called the \emph{Uniform Electron Gas} (UEG), so that the approximation becomes exact when $\rho$ is constant over a very large domain. Because the UEG is an infinite system it should not depend whether it is defined canonically or grand-canonically. Hence the function $f$ must be the same for the three functionals $F_{\rm LL}$, $F_{\rm L}$ and $F_{\rm GC}$.

The idea behind the LDA is as depicted in Figure~\ref{fig:LDA}. After subtraction of the Hartree term, one splits the space into small boxes (of volume $d\br$) and assumes that the remaining energy is the sum of the local energies. In each little box, one replaces the density by a constant. One does not use the energy of the constant function in the small box, but rather the energy per unit volume of an infinite system having the corresponding uniform density, multiplied by the volume $d\br$ of the small box.

In the LDA the complicated Levy-Lieb functionals therefore get replaced by a new universal function $f:\R_+\to\R$, which is much simpler since it only depends on one real parameter. But the function $f$ is also not known exactly, and we will see that it displays a very rich structure.

In this section we report on the results in~\cite{LewLieSei-19} where the LDA was rigorously justified for the first time. The proper regime is that of slowly varying densities, that is, densities $\rho$ which are very flat on sufficiently large domains such that $f(\rho)$ becomes a good local approximation. To this end we start by defining the function $f$.

\begin{figure}[t]
\centering
\includegraphics[width=9cm]{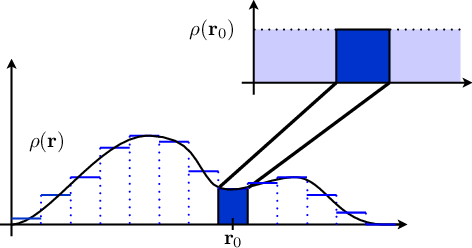}
\caption{Main idea of the Local Density Approximation in DFT. The Levy-Lieb energy (with the Hartree term subtracted) is replaced by the sum of the energies per unit volume of an infinite uniform gas with the local density $\rho(\br_0)$, times the volume $d\br$.\label{fig:LDA}}
\end{figure}

\subsection{The Uniform Electron Gas}\label{sec:UEG}

The uniform electron gas was rigorously defined in~\cite{LewLieSei-18,LewLieSei-19} and it is obtained by assuming that the density is exactly constant over a large domain which grows such as to cover the whole space. The result is the following.

\begin{theorem}[Uniform Electron Gas~\cite{LewLieSei-18,LewLieSei-19}]\label{thm:UEG}
Let $w(\br)=|\br|^{-1}$ in dimension $d=3$ and $q\geq1$ the number of spin states. Let $\rho_0>0$. Let $\Omega$ be a fixed open convex set of unit volume $|\Omega|=1$. Let $\chi\in L^1(\R^3)$ be a radial non-negative function of compact support such that $\int_{\R^3}\chi=1$ and  $\int_{\R^3}|\nabla\sqrt\chi|^2<\ii$. Then the following thermodynamic limit exists
\begin{multline}
f\big(\rho_0\big):=
\lim_{L\to\ii}L^{-3}\bigg(F_{\rm GC}\big[\rho_0\1_{L\Omega}\ast\chi\big]\\-\frac{\rho_0^2}2\iint_{\R^3\times\R^3}\frac{(\1_{L\Omega}\ast\chi)(\br)(\1_{L\Omega}\ast\chi)(\br')}{|\br-\br'|}\,\rd\br\,\rd\br'\bigg)
 \label{eq:thermodynamic_limit_quantum}
\end{multline}
and does not depend on $\Omega$ and~$\chi$.
\end{theorem}

Electrons have $q=2$ spin states but we wrote the result with a general $q$ for convenience. It is expected that the exact same result holds for the other two functionals $F_{\rm LL}$ and $F_{\rm L}$, with the same limit $f(\rho_0)$. This has not yet been proved, except in the classical case where the kinetic energy is dropped (see Section~\ref{sec:classical}).

The number $f(\rho_0)$ is the energy per unit volume (with the Hartree energy subtracted) of an infinite gas submitted to the constraint that its density is exactly constant over $\R^3$, $\rho(\br)\equiv\rho_0$. In the literature, $f(\rho_0)$ is often confused with the corresponding \emph{Jellium energy}. In Jellium there is no constraint on $\rho$ but one adds instead a uniform background of density $\rho_0$ which compensates the long range of the Coulomb potential~\cite{LieNar-75,Lewin-22}. That the two models coincide has only been shown in the classical case~\cite{CotPet-19b,LewLieSei-19b} so far, as we will mention in Section~\ref{sec:classical} below. In the quantum case the same result should hold, but the proof has not been written yet.

Here are some rigorously known properties of the function $f$.

\begin{theorem}[Properties of $f$~\cite{LieNar-75,GraSol-94,LewLieSei-18,LewLieSei-19}]\label{thm:prop_f}
The function $f$ is locally Lipschitz: There exists a constant $C$ such that
\begin{equation}
 \left|f(\rho_1)-f(\rho_2)\right|\leq C\left(\max(\rho_1,\rho_2)^{\frac13}+\max(\rho_1,\rho_2)^{\frac23}\right)|\rho_1-\rho_2|,
 \label{eq:Lipschitz}
\end{equation}
for all $\rho_1,\rho_2\geq0$. The function $f$ satisfies the uniform bound
\begin{equation}
 c_{\rm TF}(3)q^{-\frac23}\rho^{\frac53}-\frac35 \left(\frac{9\pi}{2}\right)^{\frac13}\rho^{\frac43}\leq f(\rho)\leq c_{\rm TF}(3)q^{-\frac23}\rho^{\frac53}-c_{\rm D}(1,3)q^{-\frac13}\rho^{\frac43}
 \label{eq:f_uniform_bound}
\end{equation}
for all $\rho\geq0$, where $c_{\rm TF}(3)=(3/10)(6\pi^2)^{2/3}$ and $c_{\rm D}(1,3)=(3/4)\left(6/\pi\right)^{1/3}$ are respectively the Thomas-Fermi and Dirac constants, discussed later in Sections~\ref{sec:kinetic} and~\ref{sec:classical}.
It behaves at small densities like
\begin{equation}
 f(\rho)=c_{\rm UEG}(1,3)\rho^{\frac43}+o\left(\rho^{\frac43}\right)_{\rho\to0^+}
\label{eq:f_small_rho}
 \end{equation}
where
$$ -1.4508\simeq -\frac35 \left(\frac{9\pi}{2}\right)^{\frac13} \leq c_{\rm UEG}(1,3)\leq -1.4442$$
is the classical UEG energy discussed later in Section~\ref{sec:classical}, and at large densities like
\begin{equation}
 f(\rho)= c_{\rm TF}(3)q^{-\frac23}\rho^{\frac53}-c_{\rm D}(1,3)q^{-\frac13}\rho^{\frac43}+o\left(\rho^{\frac43}\right)_{\rho\to\ii}.
 \label{eq:f_large_rho}
\end{equation}
\end{theorem}

The statement involves several constants that will be introduced in the next sections. It is believed that $f$ is smooth except at finitely many points corresponding to phase transitions. In the case of spin-$1/2$ particles like electrons ($q=2$), numerical simulations in~\cite{BruSahTel-66,PolHan-73,AldCep-80,JonCep-96,ZonLinCep-02,DruRadTraTowNee-04,HolMor-20,AzaDru-22_ppt} indicate that there might be one or two such points. For a long time it was believed that the system can be a ferromagnetic Wigner crystal, a ferromagnetic fluid and a paramagnetic fluid. Recent results indicate that the ferromagnetic fluid phase might not exist~\cite{HolMor-20,AzaDru-22_ppt}, however.
More transitions could occur in the solid phase (for instance an anti-ferromagnetic crystal).
In spite of the clear numerical evidence that there are phase transitions, proving it remains a very challenging open problem~\cite{BlaLew-15}.
Several approximate formulas for the function $f$ are used in DFT, including for instance the celebrated Perdew-Wang (PW92) functional~\cite{PerWan-92}.

\subsection{The Local Density Approximation of $F_{\rm GC}$}

We now state a result from~\cite{LewLieSei-19} where the LDA was justified for the first time in the quantum case.

\begin{theorem}[LDA for $F_{\rm GC}$~\cite{LewLieSei-19}]\label{thm:LDA}
Let $w(\br)=|\br|^{-1}$ in dimension $d=3$ and $q\geq1$ the number of spin states. Then there exists a constant $C=C(q)$ such that
\begin{multline}
 \left|F_{\rm GC}[\rho]- \frac12 \iint_{\R^3\times\R^3}\frac{\rho(\br)\rho(\br')}{|\br-\br'|}\,\rd\br\,\rd\br'-\int_{\R^3}f\big(\rho(\br)\big)\,\rd\br\right|\\
 \leq \eps \int_{\R^3}\big(\rho(\br)+\rho(\br)^2\big)\,\rd\br
 +\frac{C(1+\eps)}{\eps}\int_{\R^3}|\nabla\sqrt\rho(\br)|^2\,\rd\br+\frac{C}{\eps^{15}}\int_{\R^3}|\nabla\sqrt\rho(\br)|^4\,\rd\br
 \label{eq:LDA_main_estim}
\end{multline}
for every $\eps>0$ and every non-negative density $\rho\in L^1(\R^3)\cap L^2(\R^3)$ such that $\nabla \sqrt{\rho}\in L^2\cap L^4(\R^3)$. Here $f$ is the function defined in Theorem~\ref{thm:UEG}.
\end{theorem}

It is expected that the exact same result holds for the canonical functionals $F_{\rm LL}$ and $F_{\rm L}$, with of course the additional constraint that $\int_{\R^3}\rho\in\N$.

The last gradient term $\eps^{-15}|\nabla\sqrt\rho|^4$ was chosen for simplicity but the same result actually holds with  $\eps^{1-4p}|\nabla\rho^\theta|^p$ instead, under the conditions that $p>3$, $0<\theta<1$ and $2\leq p\theta\leq 1+{p}/{2}$. The constant $C$ then depends on the chosen $p$ and $\theta$.

In addition to the large power of $\eps$, which is an artifact of the proof in~\cite{LewLieSei-19}, the form of the error term is probably not optimal. It is reasonable to expect that the right side of~\eqref{eq:LDA_main_estim} should only involve quantities like $\rho^{5/3}$, $\rho^{4/3}$, $|\nabla\sqrt\rho|^2$ and $|\nabla\rho^{1/3}|^2$ or perhaps $|\nabla\rho|$ which have the same scaling as the kinetic and Coulomb terms.

The inequality~\eqref{eq:LDA_main_estim} holds for every density but it is useful only when the two gradient terms are much smaller than the
first term,
$$ \int_{\R^3}\left(|\nabla\sqrt{\rho}(\br)|^2+|\nabla\sqrt{\rho}(\br)|^4\right)\,\rd\br \ll \int_{\R^3}\left(\rho(\br)+\rho(\br)^2\right)d\br,$$
so that after optimizing over $\eps$ one gets a small term on the right side of~\eqref{eq:LDA_main_estim}. One interesting case is when the density is given in terms of a fixed function $\rho$ with $\int_{\R^3}\rho=1$, which is rescaled in the manner
$$\rho_N(\br)=\rho(N^{-1/3}\br).$$
After optimizing over $\eps$ we obtain the following expansion  of the grand-canonical Levy-Lieb energy:
\begin{multline}
 F_{\rm GC}[\rho_N]= \frac{N^{\frac53}}2 \iint_{\R^3\times\R^3}\frac{\rho(\br)\rho(\br')}{|\br-\br'|}\,\rd\br\,\rd\br'+N\int_{\R^3}f\big(\rho(\br)\big)\,\rd\br +O\left(N^{\frac{11}{12}}\right).
 \label{eq:LDA_main_rescaled_density}
\end{multline}
The first term is the trivial non-local Coulomb term, whereas the next term in the expansion is the LDA. It is an interesting open question to determine the next order correction, which is believed to be also local, of order $N^{1/3}$ and to involve gradients. The exact same result as~\eqref{eq:LDA_main_rescaled_density} is expected for $F_{\rm LL}$ and $F_{\rm L}$.

\begin{remark}[LDA for the exchange-correlation energy]\rm
Theorem~\ref{thm:LDA_kinetic} below is a result similar to Theorem~\ref{thm:LDA} for the grand-canonical kinetic energy  $T_{\rm GC}$ alone, and implies a corresponding bound for the difference of the two functionals. These two bounds justifies the LDA for the (grand-canonical) exchange-correlation energy, as was defined in Section~\ref{sec:KS}.
\end{remark}

In the next two sections we study separately the kinetic energy functional and the classical interaction functional. We discuss known upper and lower bounds and derive the LDA for these functionals in a similar (but simpler) manner as for the full Levy-Lieb functional $F_{\rm GC}$. Although the minimum of a sum is in general not the sum of the two minima, understanding the kinetic and interaction energies separately will give us useful information on the full functional, as explained in Section~\ref{sec:TFD_is_a_lower_bound} below.

\section{Kinetic energy and Lieb-Thirring inequalities}\label{sec:kinetic}

\subsection{Three kinetic energy functionals}
We have introduced in~\eqref{def:Slater_kinetic_energy} the lowest kinetic energy $T_{\rm S}[\rho]$ that can be reached with Slater determinants, for a given density $\rho$ with $\int_{\R^3}\rho(\br)\,\rd\br=N$. We can define in a similar manner the lowest kinetic energy that can be reached with all possible wave functions
\begin{equation}
\boxed{T[\rho]:=\min_{\substack{\Psi\in \bigwedge_1^NH^1(\R^d\times\Z_q,\C)\\ \norm{\Psi}_{L^2}=1\\ \rho_\Psi=\rho}}\frac12\sum_{j=1}^N\int_{(\R^d\times\Z_q)^N}|\nabla_{j}\Psi(\bX)|^2\,\rd\bX.}
\label{def:kinetic_energy}
\end{equation}
This is nothing but $F_{\rm LL}^0[\rho]$, the Levy-Lieb functional with interaction $w\equiv0$. Recall that $T$ and $T_{\rm S}$ depend on the number of spin states $q$.

Since this is a non-interacting problem, one may think at first sight that minimizers will always be Slater determinants, that is, $T[\rho]$ and $T_{\rm S}[\rho]$ should coincide. But this is not true in general~\cite{Lieb-83b} and the best one can say for a general $\rho$ is that $T_{\rm S}[\rho]\geq T[\rho]$.

There are two other natural kinetic functionals corresponding to $F_{\rm L}^0[\rho]$ and $F_{\rm GC}^0[\rho]$, respectively. For the first one the minimization is extended to mixed canonical states and for the second one to grand-canonical states. It turns out that these two are equal:
$$F_{\rm L}^0[\rho]=F_{\rm GC}^0[\rho].$$
The reason is that the kinetic energy can be expressed in terms of the one-particle density matrix $\gamma$ and that the set of such matrices which are $N$-representable by a mixed state coincides with those which are representable by a grand-canonical state~\cite{Coleman-63,ColYuk-00}. By duality, this also follows from the fact that the inequality~\eqref{eq:conjecture_convexity_E_n} always holds in the non-interacting case.

In order to explain all this in detail, we first recall that the one-particle density matrix $\gamma_\Psi$ of a wavefunction $\Psi$ is the self-adjoint operator acting on the one-particle space $L^2(\R^d\times\Z_q)$ with integral kernel
$$\gamma_\Psi(\bx,\by)=N\int_{(\R^d\times \Z_q)^{N-1}}\Psi(\bx,\bX)\overline{\Psi(\by,\bX)}\,\rd\bX.$$
This gives
$$\frac12\sum_{j=1}^N\int_{(\R^d\times\Z_q)^N}|\nabla_{j}\Psi(\bX)|^2\,\rd\bX=\tr\left(\frac{-\Delta}{2}\right)\gamma_\Psi$$
with the trace interpreted in the quadratic form sense. Every density matrix of an antisymmetric $\Psi$ satisfies $0\leq \gamma_\Psi=(\gamma_\Psi)^*\leq1$ and $\tr(\gamma_\Psi)=N$. When we consider mixed $N$-particle states we obtain the convex hull of the set of $N$-representable density matrices. This convex hull is definitely contained in the convex set
$$\left\{\gamma=\gamma^*\ :\ 0\leq\gamma\leq 1,\ \tr(-\Delta)\gamma<\ii,\ \tr(\gamma)=N\right\}.$$
But the extreme points of this set are the rank--$N$ orthogonal projections with finite kinetic energy. Those are exactly the one-particle density matrices of the Slater determinants. Hence we must have equality of the two convex sets. By considering grand-canonical states the set will not increase further. See Theorem~\ref{thm:upper_bound_F_L} for a related result.

This discussion leads us towards introducing the following kinetic energy functional
\begin{equation}
\boxed{T_{\rm GC}[\rho]:=\min_{\substack{0\leq\gamma=\gamma^*\leq1\\ \tr(-\Delta)\gamma<\ii\\ \rho_\gamma=\rho}}\tr\left(\frac{-\Delta}{2}\right)\gamma.}
\label{def:kinetic_GC}
\end{equation}
We call it ``grand-canonical'' since $\int_{\R^d}\rho$ can now take any positive value. But for $\int_{\R^d}\rho\in\N$, this is just Lieb's canonical energy $T_{\rm GC}[\rho]=F_{\rm L}^0[\rho]$. In this case we also have
$$T_{\rm GC}[\rho]\leq T[\rho]\leq T_{\rm S}[\rho],\qquad \text{when $\int_{\R^d}\rho\in\N$}.$$
The functional $T_{\rm GC}[\rho]$ is convex and it is the convex hull and the weak-$\ast$ semi-continuous closure of both $T[\rho]$ and $T_{\rm S}[\rho]$, similarly as in Theorems~\ref{thm:wlsc_F_GC}. Minimizers exist for these three functionals, as in Theorem~\ref{thm:exists_min}. It suffices to take $w\equiv0$ in all those theorems.

\begin{remark}[$N$-representability of the one-particle density matrix]\rm
There are complicated constraints on a one-particle density matrix $\gamma$ to ensure that it arises from an $N$-particle wavefunction $\Psi$. For instance, when $N=2$ then all the eigenvalues of $\gamma$ must be of even multiplicity. Another example is that no $\gamma$ of rank $N+1$ is $N$-representable. See~\cite{ColYuk-00,Friesecke-03,Lewin-04a} for these two examples and~\cite{Ruskai-70,BorDen-72,Klyachko-06,AltKly-08} for more advanced results when $N\geq3$.
\end{remark}

Now we mention some known upper and lower bounds on the kinetic energy functionals. Lower bounds naturally involve the lowest functional $T_{\rm GC}[\rho]$. Upper bound should ideally involve $T_{\rm S}[\rho]$ but we will see that much more is known on $T_{\rm GC}[\rho]$.

\subsection{Lower bounds: Hoffmann-Ostenhof and Lieb-Thirring inequalities}\label{sec:LT}

The first lower bound is the Hoffmann-Ostenhof inequality mentioned previously in~\eqref{eq:Hoffmann-Ostenhof} and which holds for the grand-canonical kinetic energy as well.

\begin{theorem}[Hoffmann-Ostenhof inequality~\cite{Hof-77}]\label{thm:HO}
For every $\rho\geq0$ such that $\sqrt\rho\in H^1(\R^d)$, we have
\begin{equation}
T_{\rm GC}[\rho]\geq  \frac12 \int_{\R^d}\left|\nabla\sqrt{\rho}(\br)\right|^2\,\rd\br.
 \label{eq:Hoffmann-Ostenhof_GC}
\end{equation}
\end{theorem}

Using the Gagliardo-Nirenberg inequality~\cite{LieLos-01}
\begin{equation}
\norm{u}_{L^2(\R^d)}^{\frac4d}\int_{\R^d}|\nabla u(\br)|^2\,\rd\br\geq c_{\rm GN}(d)\int_{\R^d}|u(\br)|^{2+\frac4d}\,\rd\br,
\label{eq:Gagliardo-Nirenberg}
\end{equation}
for $u=\sqrt\rho$, we obtain
$$T_{\rm GC}[\rho]\geq \frac{c_{\rm GN}(d)}{2N^{\frac2d}}\int_{\R^d}\rho(\br)^{1+\frac2d}\,\rd\br$$
and this is optimal for bosons. But for fermions this is not optimal at all. The Lieb-Thirring inequality states that one can replace the $N$-dependent prefactor by an $N$-independent one (or, rather, by a $q$-dependent constant where $q$ is the number of spin states).

\begin{theorem}[Lieb-Thirring{~\cite{LieThi-75,LieThi-76,LieSei-09}}]
There exists a positive constant $c_{\rm LT}(d)>0$ such that
\begin{equation}
T_{\rm GC}[\rho]\geq q^{-\frac2d}c_{\rm LT}(d)\int_{\R^d} \rho(\br)^{1+\frac2d}\,\rd\br
\label{eq:Lieb-Thirring}
\end{equation}
for all $\rho\geq0$ such that $\sqrt\rho\in H^1(\R^d)$.
\end{theorem}

Note the spin dependence in $q^{-2/d}$ which is compatible with the fact that the bosonic case is recovered when $q=N$. For particles like electrons we have $q=2$ and the constant is $N$-independent.

For large fermionic systems the Lieb-Thirring inequality is an advantageous replacement for the Gagliardo-Nirenberg inequality, to which it reduces in the case $N=1$ (in particular we always have $c_{\rm LT}(d)\leq c_{\rm GN}(d)/2$).
Since its invention, the Lieb-Thirring inequality~\eqref{eq:Lieb-Thirring} has played a central role in the mathematical understanding of large fermionic systems. It was originally used to give a proof of stability of matter~\cite{LieThi-75,Lieb-76,Lieb-90,LieSei-09} that is much shorter than the original proof of Dyson and Lenard~\cite{DysLen-67}. Later the Lieb-Thirring inequality was generalized to systems at positive density~\cite{FraLewLieSei-11,FraLewLieSei-12} where $\rho$ is a local perturbation of a constant, and to the dynamic case where it extends Strichartz's inequality~\cite{FraLewLieSei-14,FraSab-17}.

The right side of~\eqref{eq:Lieb-Thirring} is related to the kinetic energy of the free Fermi gas. Indeed, we recall that the translation-invariant orthogonal projector
\begin{equation}
 P_{\rho_0}=\1\left(-\Delta\leq 2\frac{d+2}d c_{\rm TF}(d)q^{-\frac2d}\rho_0^{\frac2d}\right)
 \label{eq:free_Fermi_gas}
\end{equation}
has the constant density $\rho_{P_{\rho_0}}=\rho_0$ and the constant kinetic energy density $c_{\rm TF}(d)\,q^{-2/d}\rho_0^{1+2/d}$, where
\begin{equation}
c_{\rm TF}(d)=\frac{2\pi^2 d}{(d+2)} \left( \dfrac{d}{| \mathbb{S}^{d-1} |} \right)^{\frac2d}
 \label{eq:c_TF}
\end{equation}
is called the \emph{Thomas-Fermi} constant.

The best constant $c_{\rm LT}(d)$ in~\eqref{eq:Lieb-Thirring} is unknown but it is definitely less or equal to $c_{\rm TF}(d)$. This is seen by using the trial state $\gamma_R=\chi(\cdot/R)P_{\rho_0}\chi(\cdot/R)$ and taking the limit $R\to\ii$.
The famous \emph{Lieb-Thirring conjecture}~\cite{Schimmer-22} states that
\begin{equation}
c_{\rm LT}(d)=\min\left\{\frac{c_{\rm GN}(d)}{2},c_{\rm TF}(d)\right\}=
\begin{cases}
\dps \frac{c_{\rm GN}(d)}{2}&\text{for $d=1,2$,}\\[0.3cm]
c_{\rm TF}(d)&\text{for $d\geq3$.}
\end{cases}
\label{eq:LT_conjecture}
\end{equation}
In other words, the conjecture states that the best constant is obtained either for the infinite non-interacting uniform electron gas, or for one isolated electron. This conjecture was investigated numerically in~\cite{Levitt-14}.
The proof of the conjecture~\eqref{eq:LT_conjecture} in dimension $d=3$ would have a great impact since it would mean that the Thomas-Fermi-Dirac (TFD) energy is an exact lower bound to the many-particle problem~\cite{Lieb-81b}, as we will mention later in Section~\ref{sec:TFD_is_a_lower_bound}. The Thomas-Fermi energy is the simplest functional in Density Functional Theory and knowing that it is an exact lower bound would simplify drastically many mathematical results, in addition to increasing its physical significance.

The best known estimate on $c_{\rm LT}(d)$ was recently proved in~\cite{FraHunJexNam-21} and reads
$$\forall d\geq1,\qquad \frac{c_{\rm LT}(d)}{c_{\rm TF}(d)}\geq (1.456)^{-\frac2d}.$$
It improves upon the previously best known result where 1.456 was replaced by 1.814 and which was proved in $d=1$ in 1991 by Eden and Foias~\cite{EdeFoi-91} and in $d\geq2$ by Dolbeault, Laptev and Loss~\cite{DolLapLos-08} in 2008. We refer to~\cite{FraHunJexNam-21,Frank-21,Schimmer-22,FraLapWei-LT} for a recent overview of other important results on the Lieb-Thirring inequality.

By duality, the Lieb-Thirring inequality implies a bound on the sum of the negative eigenvalues of a one-particle Schr\"odinger operator in an external potential $v$, denoted by functional calculus as $-\tr(-\Delta/2+v)_-$. Namely, we have by~\eqref{eq:relation_E_F_LL}
\begin{align}
-\tr(-\Delta/2+v)_-&=\inf_{N\geq0}E_N^0[v]\nn\\
&=\inf_{\substack{\sqrt\rho\in H^1(\R^d)\\\int_{\R^d}v_+\rho<\ii}}\left\{T_{\rm GC}[\rho]+\int_{\R^d}v(\br)\,\rho(\br)\,\rd\br\right\}\nn\\
&\geq \inf_{\substack{\sqrt\rho\in H^1(\R^d)\\ \int_{\R^d}v_+\rho<\ii}}\left\{q^{-\frac2d}c_{\rm LT}(d)\int_{\R^d}\rho(\br)^{1+\frac2d}\,\rd\br+\int_{\R^d}v(\br)\,\rho(\br)\,\rd\br\right\}\nn\\
&=-\frac{2 d^{\frac{d}2}q}{(d+2)^{1+\frac{d}2}c_{\rm LT}(d)^{\frac{d}2}}\int_{\R^d}v_-(\br)^{1+\frac{d}2}\,\rd\br\label{eq:Lieb-Thirring_dual_v}
\end{align}
where $E_N^0$ indicates that we take $w\equiv0$. Since $T_{\rm GC}[\rho]=F^0_{\rm GC}[\rho]$ is the Legendre transform of $E_N^0[v]$, the inequality~\eqref{eq:Lieb-Thirring_dual_v} is actually equivalent to the Lieb-Thirring inequality~\eqref{eq:Lieb-Thirring}. The original proof of Lieb and Thirring~\cite{LieThi-75,LieThi-76} was actually showing~\eqref{eq:Lieb-Thirring_dual_v} and it is only much later that Rumin~\cite{Rumin-11} found a direct proof of~\eqref{eq:Lieb-Thirring}.

The semi-classical constant~\eqref{eq:c_TF} naturally occurs for slowly varying densities in the LDA regime, as we will see. Nam  proved in~\cite{Nam-18} that one can replace the Lieb-Thirring (unknown) constant $c_{\rm LT}(d)$ by
$c_{\rm TF}(d)$ at the expense of a gradient correction.

\begin{theorem}[Nam's Lieb-Thirring inequality with gradient correction~\cite{Nam-18}]
Let $q,d\geq1$. There exists a universal constant $\kappa(d)$ (independent of the number of spin states $q$) such that
\begin{equation}
T_{\rm GC}[\rho]\geq q^{-\frac2d}c_{\rm TF}(d)(1-\eps)\int_{\R^d} \rho(\br)^{1+\frac2d}\,\rd\br-\frac{\kappa(d)}{\eps^{3+\frac4d}}\int_{\R^d}|\nabla\sqrt{\rho}(\br)|^2\,\rd\br
\label{eq:Lieb-Thirring_Nam}
\end{equation}
for all $0<\eps<1$ and all $d\geq1$.
\end{theorem}

This was the first step towards a proof of the validity of the LDA for the kinetic energy, to which we will come back soon.

Li and Yau  proved in~\cite{LiYau-83}  a lower bound involving the optimal Thomas-Fermi constant in a bounded domain $\Omega\subset\R^d$:
\begin{equation}
T_{\rm GC}[\rho]\geq c_{\rm TF}(d)|\Omega|^{-\frac23}\left(\int_\Omega\rho(\br)\,\rd\br\right)^{\frac53}\qquad\text{for all $\sqrt\rho\in H^1_0(\Omega)$}
\label{eq:LiYau}
\end{equation}
(see also~\cite[Thm.~12.3]{LieLos-01} and~\cite[Lem.~9]{HaiLewSol_2-09}).
The bound is particularly useful for densities $\rho$ which are (almost) constant over a domain $\Omega$. For instance for $\rho(\br)=\rho_0\1_\Omega\ast\chi(\br)$ with ${\rm supp}(\chi)\subset B_1$ we find the exact lower bound
$$T_{\rm GC}[\rho_0\1_\Omega\ast\chi(\br)]\geq c_{\rm TF}(d)(\rho_0)^{\frac53}\frac{|\Omega|^{\frac53}}{|\Omega+B|^{\frac23}}=c_{\rm TF}(d)(\rho_0)^{\frac53}\left(|\Omega|-C|\Omega|^{\frac23}\right)_+$$
whereas~\eqref{eq:Lieb-Thirring_Nam} yields a worse error term.

\subsection{Upper bounds}

For upper bounds one should ideally consider the larger functional $T_{\rm S}[\rho]$. In dimension $d=1$, March and Young~\cite[Eq. (9)]{MarYou-58}  gave the proof of an estimate similar to~\eqref{eq:Lieb-Thirring_Nam} without the parameter $\eps$ in front of the gradient correction
\begin{equation}
 T_{\rm S}[\rho]\leq q^{-2}\frac{\pi^2}{6}\int_\R \rho(x)^3\,\rd x+\frac12\int_\R\left|\big(\sqrt{\rho}\big)'(x)\right|^2\,\rd x
 \label{eq:March-Young}
\end{equation}
where $\pi^2/6=c_{\rm TF}(1)$. In the same paper they also state a result in 3D (for a constant $c>c_{\rm TF}(3)$) but the proof has a mistake~\cite[Sec. 5.B]{Lieb-83b}.
The bound~\eqref{eq:March-Young} is proved by using as trial state the orbitals
\begin{equation}
 \phi_n(x)=\sqrt{\frac{\rho(x)}{N}}\exp\left(\frac{2in\pi}{N} \int_{-\ii}^x\rho(t)\,\rd t\right)
 \label{eq:March-Young-orbitals}
\end{equation}
(we take $q=1$ for simplicity), where $n\in\Z$ and the phases are seen to make the $\phi_n$ orthonormal. Computing the kinetic energy of this trial state, one obtains
$$T_{\rm S}[\rho]\leq 2\pi^2\frac{\sum n^2}{N^3}\int_\R \rho(x)^3\,\rd x+\frac12\int_\R\left|\big(\sqrt{\rho}\big)'(x)\right|^2\,\rd x.$$
Taking all the integers $n$ less than or equal to $N/2$ in absolute value and using the precise behavior of the series gives the result for $q=1$.

The method can be generalized to higher dimensions using a similar method, but the estimate has a bad behavior in $N$. The orbitals
\begin{equation}
\phi_n(\bx)=\sqrt{\frac{\rho(\br)}{N}}e^{i\theta_n(\br)}\delta_{0}(\sigma),\qquad \theta_n(\br)=\frac{2n\pi}{N} \int_{-\ii}^{r_1}\int_{\R^{d-1}}\rho(t,\br')\,\rd t \,\rd\br'
\label{eq:phases}
\end{equation}
were considered in~\cite{Harriman-81,Lieb-83b} and these are the phases which we  already mentioned in~\eqref{eq:phases_pre}. Using this trial state one obtains~\cite{Lieb-83b} for $q=1$
\begin{equation}
 T_{\rm S}[\rho]\leq \left(\frac{2\pi^2}{3}N^2+CN\right)\int_{\R^d}|\nabla\sqrt{\rho}(\br)|^2\,\rd\br.
 \label{eq:bound_T_Lieb}
\end{equation}
An upper bound on $T_{\rm S}[\rho]$ involving only $\int_{\R^d}|\nabla\sqrt\rho|^2$ has to have a constant diverging at least as fast as $N^{2/d}$, due to the Lieb-Thirring inequality. In~\cite{ZumMas-83,ZumMas-84,BokGre-96} the optimal upper bound of this form was shown:
\begin{equation}
 T_{\rm S}[\rho]\leq CN^{\frac2d}\int_{\R^d}|\nabla\sqrt{\rho}(\br)|^2\,\rd\br.
 \label{eq:bound_T_Grebert}
\end{equation}
The idea of the proof is to apply a deformation of the space in order to map $\rho$ onto the constant density in a box, which is then represented by a usual Slater determinant made of plane waves. One would expect an upper bound on $T_{\rm S}[\rho]$ involving both $\int_{\R^d}\rho^{1+2/d}$ and $\int_{\R^d}|\nabla\sqrt\rho|^2$, with coefficients independent of $N$ as in~\eqref{eq:March-Young} but this seems unknown at present. The periodic case was studied in~\cite{BokGreMau-03}.

Recently, an upper bound similar to~\eqref{eq:Lieb-Thirring_Nam} was proved in~\cite{LewLieSei-19} for the grand-canonical functional $T_{\rm GC}[\rho]$.

\begin{theorem}[Upper bound on {$T_{\rm GC}[\rho]$~\cite{LewLieSei-19}}]\label{thm:upper_bound_T}
Let $d,q\geq1$. There exists a constant $\kappa'(d)$ such that
\begin{equation}
 T_{\rm GC}[\rho]\leq q^{-\frac2d}c_{\rm TF}(d)\left(1+\eps\right)\int_{\R^d}\rho(\br)^{1+\frac2d}\,\rd\br
 +\kappa'(d)\frac{1+\eps}{\eps}\int_{\R^d}|\nabla\sqrt{\rho}(\br)|^2\,\rd\br
 \label{eq:upper_bound_T_GC}
\end{equation}
for all $\eps>0$ and all $\rho\geq0$ with $\sqrt\rho\in H^1(\R^d)$.
\end{theorem}

The main difficulty in the proof of~\eqref{eq:upper_bound_T_GC} is the constraint that the one-particle density matrix must have the exact density $\rho$. One can provide rather good upper bounds if we allow to vary the density a bit. For instance, by using coherent states~\cite{Lieb-81b} the density $\rho$ is replaced by $\rho\ast|f|^2$ where $f$ is the profile used to build the coherent states (typically a Gaussian).

The proof of~\eqref{eq:upper_bound_T_GC} instead relies on the following trial  one-particle density matrix
\begin{equation}
\gamma=\int_0^\ii  \sqrt{\eta\left(\frac{t}{\rho(\br)}\right)}\;\1\left(-\Delta\leq 2\frac{d+2}d c_{\rm TF}(d)\,q^{-\frac2d}t^{\frac2d}\right)\;\sqrt{\eta\left(\frac{t}{\rho(\br)}\right)}\;\frac{\rd t}{t}.
\label{eq:trial_state}
\end{equation}
Here the two functions $\sqrt{\eta(t/\rho(\br))}$ are interpreted as multiplication operators, whereas the operator in the middle is the Fourier multiplier $P_t$ introduced before in~\eqref{eq:free_Fermi_gas}. The non-negative function $\eta$ is chosen such that
\begin{equation}
\int_0^\ii\eta(t)\,\rd t=1,\qquad \int_0^\ii\eta(t)\,\frac{\rd t}{t}\leq1.
\label{eq:cond_eta}
\end{equation}
The main idea is to represent the density $\rho$ by using the smooth ``layer cake principle''~\cite[Thm.~1.13]{LieLos-01}
$$\rho(\br)=\int_0^\ii \eta\left(\frac{t}{\rho(\br)}\right)\,\rd t$$
where we think of $\eta$ as very concentrated around $1$, and to then take the free Fermi gas $P_t$ as in~\eqref{eq:free_Fermi_gas} on the support of $\eta(t/\rho)$ where $\rho$ is very close to $t$. The measure $\rd t/t$ in~\eqref{eq:trial_state} ensures that $\rho_\gamma=\rho$ exactly. On the other hand the condition $\int_0^\ii\eta(t)\,{\rd t}/{t}\leq1$ ensures that $0\leq\gamma\leq1$ and means that $\eta$ must put slightly more weight on the right of $1$ than on the left. Computing the kinetic energy of the trial state~\eqref{eq:trial_state} and optimizing over $\eta$, one obtains~\eqref{eq:upper_bound_T_GC}.

We have explained the construction of the trial state~\eqref{eq:trial_state} to emphasize how much easier it is to work in the grand-canonical setting. It is an important open problem to obtain a bound similar to~\eqref{eq:upper_bound_T_GC} on $T_{\rm S}[\rho]$ or $T[\rho]$. For $T_{\rm S}[\rho]$ this amounts to understanding how to build $N$ orthogonal orbitals with the prescribed density, and to obtain the lowest possible energy. This problem is somewhat related to the smooth Hobby-Rice problem. There one considers $N$ arbitrary $L^2$--normalized functions $\phi_1,...,\phi_N\in H^1(\R^d,\C)$ and looks for the minimal kinetic energy cost to orthonormalize them using only phases: $\phi_j'=\phi_j e^{i\theta_j}$. It was proved in~\cite{LazLie-13,Rutherfoord-13,FriSup-22} that such phases $\theta_j$ always exist, but known bounds involve $\norm{\nabla\theta_j}_{L^1}$ which are not enough to deduce anything on the $H^1$ norm of the orbitals $\phi_j'$. In view of~\eqref{eq:bound_T_Lieb}, one would suspect that
$$\min_{\substack{|\phi_j'|=|\phi_j|\\\pscal{\phi_j',\phi_k'}=\delta_{jk}}}\sum_{j=1}^N\int_{\R^d\times\Z_q}|\nabla\phi_j'(\bx)|^2\,\rd\bx\leq C(N,d)\sum_{j=1}^N\int_{\R^d\times\Z_q}|\nabla|\phi_j|(\bx)|^2\,\rd\bx$$
but this seems unknown at present. In~\eqref{eq:March-Young-orbitals} the reference orbitals are all equal to $\sqrt{\rho/N}$ but this is probably not the optimal choice for $T_{\rm S}[\rho]$ in dimension $d\geq2$.

\subsection{Local Density Approximation for the kinetic energy}

From the lower bound~\eqref{eq:Lieb-Thirring_Nam} and the upper bound~\eqref{eq:upper_bound_T_GC} we obtain the following result which is similar to Theorem~\ref{thm:LDA} but involves only quantities that all scale the same, namely like inverse-length squared.

\begin{theorem}[Local Density Approximation of the kinetic energy~{\cite{Nam-18,LewLieSei-19}}]\label{thm:LDA_kinetic}
Let $d,q\geq1$. There exists a universal constant $C(d)$ such that
\begin{multline}
\left|T_{\rm GC}[\rho]-q^{-\frac2d}c_{\rm TF}(d)\int_{\R^d}\rho(\br)^{1+\frac2d}\rd\br\right|\\
\leq \eps q^{-\frac2d}\int_{\R^d}\rho(\br)^{1+\frac2d}\rd\br+C(d)\left(1+\eps^{-3-\frac4d}\right)\int_{\R^d}|\nabla\sqrt{\rho}(\br)|^2\,\rd\br
\label{eq:LDA_kinetic_bound}
\end{multline}
for all $\rho\geq0$ with $\sqrt\rho\in H^1(\R^d)$ and all $\eps>0$.
\end{theorem}

In the regime where
$$ \int_{\R^d}|\nabla\sqrt{\rho}(\br)|^2\,\rd\br \ll \int_{\R^d}\rho(\br)^{1+\frac2d}\rd\br$$
the optimization over $\eps$ gives a right side which is negligible compared to the left side. In this regime we can approximate the kinetic energy functional in the manner
\begin{equation}
T_{\rm GC}[\rho]\approx q^{-\frac2d}c_{\rm TF}(d)\int_{\R^d}\rho(\br)^{1+\frac2d}\rd\br.
\label{eq:LDA_kinetic}
\end{equation}
The right side is called the Thomas-Fermi kinetic energy and it is the simplest approximation to $T_{\rm GC}[\rho]$. If we fix a density $\rho$ with $\int_{\R^d}\rho=1$ and take $\rho_N(\br)=\rho(\br N^{-1/d})$, then we find from~\eqref{eq:LDA_kinetic_bound} that
$$T_{\rm GC}[\rho(\cdot N^{-1/d})]=Nq^{-\frac2d}c_{\rm TF}(d)\int_{\R^d}\rho(\br)^{1+\frac2d}\rd\br+O\left(N^{\frac{2d+1}{2d+2}}\right).$$
From semi-classical analysis it is expected that for a sufficiently regular $\rho$ the next term should be equal to
\begin{equation}
\frac{d-2}{6d}N^{1-\frac2d}\int_{\R^d}|\nabla\sqrt{\rho}(\br)|^2\,\rd\br.
\label{eq:Weizacker}
\end{equation}
which is called the \emph{von Weizs\"acker correction}, see~\cite[Sec.~6.7]{ParYan-94} and~\cite[p.~89--90]{LunMar-83}. This is in reference to the historical work~\cite{Weizsacker-35} for atoms in dimension $d=3$ where however von Weizs\"acker chose the coefficient $1/2$ instead of $1/18$.\footnote{In order to recover Scott's correction in atoms, the coefficient must actually be taken equal to $0.083$~\cite{Lieb-81b}.} The value of the prefactor in~\eqref{eq:Weizacker} was predicted in~\cite{Kirzhnits-57,HolKozMar-91,Salasnich-07,KoiSto-07,TraLenNgEng-17}. That the coefficient is negative in dimension $d=1$ is related to the non-optimality of the Thomas-Fermi constant in the Lieb-Thirring inequality~\eqref{eq:Lieb-Thirring}
and is well known in one-dimensional semi-classical analysis~\cite{Burke-22}.

Even without having a clean upper bound like~\eqref{eq:upper_bound_T_GC}, it is reasonable to believe that
$$\lim_{N\to\ii}\frac{T[\rho(\cdot/N^{1/d})]}{N}=\lim_{N\to\ii}\frac{T_{\rm S}[\rho(\cdot/N^{1/d})]}{N}=q^{-\frac2d}c_{\rm TF}(d)\int_{\R^d}\rho(\br)^{1+\frac2d}\rd\br$$
but this does not seem to be known at present. If the fixed density $\rho$ is replaced by a well chosen locally constant density $\rho_N$ converging to $\rho$, then this was proved in~\cite[Thm.~4]{GotNam-18}.

\subsection{Derivation from Levy-Lieb at large densities}

In this section we show that our kinetic energy functionals can be obtained from the corresponding Levy-Lieb functionals in a proper limit of large densities. For completeness, we consider a rather arbitrary interaction potential $w$ in any dimension.

\begin{theorem}[Convergence at high density]\label{thm:CV_T}
Let $w\in L^p(\R^d)+L^\ii(\R^d)$ with $p$ as in~\eqref{eq:hyp_p} and $\rho\geq0$ such that $\sqrt\rho\in H^1(\R^d)$. If $\int_{\R^d}\rho\in\N$ we have
\begin{equation}
\lim_{\lambda\to\ii}\frac{F_{\rm LL}\big[\lambda^d\rho(\lambda\,\cdot)\big]}{\lambda^2}=T[\rho],\qquad \lim_{\lambda\to\ii}\frac{F_{\rm L}\big[\lambda^d\rho(\lambda\,\cdot)\big]}{\lambda^2}=T_{\rm GC}[\rho].
\label{eq:limit_F_LL_T}
\end{equation}
If $\int_{\R^d}\rho\in\R_+$ and the additional classical stability assumption~\eqref{eq:w_class_stable} holds, we have
\begin{equation}
\lim_{\lambda\to\ii}\frac{F_{\rm GC}\big[\lambda^d\rho(\lambda\,\cdot)\big]}{\lambda^2}=T_{\rm GC}[\rho].
\label{eq:limit_F_GC_T}
\end{equation}
\end{theorem}

One can also prove the convergence of optimal states or even write the theorem in the form of Gamma convergence. In a similar manner, $T_{\rm S}[\rho]$ arises from the Hartree-Fock-type Levy-Lieb functional where one only minimizes over Slater determinants.

Since we have not found the proof in the literature, we provide it here for completeness.

\begin{proof}
We start with $F_{\rm LL}$. By scaling we see that
$$\frac{F_{\rm LL}^w\big[\lambda^d\rho(\lambda\cdot)\big]}{\lambda^2}=F_{\rm LL}^{w_\lambda}\big[\rho\big]$$
with the new interaction potential $w_\lambda(\br)=\lambda^{-2}w(\br/\lambda)$. Our assumptions on $w$ imply that $w$ is infinitesimally $(-\Delta)$-form bounded, that is, $|w|\leq \eps(-\Delta)+C_\eps$
for all $\eps>0$. After scaling this implies
$$|w_\lambda|\leq \eps(-\Delta)+\frac{C_\eps}{\lambda^2}.$$
For the two-particle operator this gives
$$\sum_{1\leq j<k\leq N}|w_\lambda(\br_j-\br_k)|=\frac12 \sum_{j=1}^N\sum_{k\neq j}|w_\lambda(\br_j-\br_k)|\leq \frac{N-1}2 \sum_{j=1}^N\left(-\eps \Delta_j+\frac{C_\eps}{\lambda^2}\right)$$
and we thus obtain
\begin{multline}
\frac{1-\eps(N-1)}2\sum_{j=1}^N(-\Delta)_{j}-\frac{N(N-1) C_\eps}{2\lambda^2}\\
\leq H_N^{0,w_\lambda}\leq \frac{1+\eps(N-1)}2\sum_{j=1}^N(-\Delta)_{j}+\frac{N(N-1) C_\eps}{2\lambda^2}. \label{eq:op_upper_lower_bound}
\end{multline}
This yields the bound
\begin{multline*}
(1-\eps(N-1))T[\rho]-\frac{N(N-1) C_\eps}{2\lambda^2}\\
\leq \frac{F_{\rm LL}^w\big[\lambda^d\rho(\lambda\cdot)\big]}{\lambda^2}\leq (1+\eps(N-1))T[\rho]+\frac{N(N-1) C_\eps}{2\lambda^2}.
\end{multline*}
The limit~\eqref{eq:limit_F_LL_T} follows after taking first $\lambda\to\ii$ and then $\eps\to0$. For an explicit potential such as Coulomb we know how $C_\eps$ depends on $\eps$ and one can then give a quantitative bound.

For $F_{\rm L}$ the argument is exactly the same, with the same bound and $T[\rho]$ replaced by $T_{\rm GC}[\rho]$. For $F_{\rm GC}$ the above argument does not work due to the bad behavior in $N$. Instead, we rescale the stability assumption~\eqref{eq:w_class_stable} on $w$ and obtain
$$H_n^{0,w_\lambda}\geq \sum_{j=1}^n(-\Delta)_{j}-\frac{C}{\lambda^2}n$$
which provides the lower bound
$$\frac{F_{\rm GC}^w\big[\lambda^d\rho(\lambda\cdot)\big]}{\lambda^2}\geq T_{\rm GC}[\rho]-\frac{C}{\lambda^2}\int_{\R^d}\rho(\br)\,\rd\br.$$
For the upper bound we consider a fixed grand-canonical state $\Gamma=(\Gamma_n)_{n\geq0}$ such that
$$\sum_{n\geq1}\tr(H^{0,0}_n\Gamma_n)\leq T_{\rm GC}[\rho]+\eta$$
for some small $\eta>0$. From the proof in Appendix~\ref{app:wlsc_F_GC}, we can assume that $\Gamma$ has compact support: $\Gamma_n\equiv0$ for $n\geq K$. Using the previous bound~\eqref{eq:op_upper_lower_bound} in the canonical case, we obtain the bound
$$\frac{F_{\rm LL}^w\big[\lambda^d\rho(\lambda\cdot)\big]}{\lambda^2}\leq (1+\eps(K-1))\big(T_{\rm GC}[\rho]+\eta\big)+\frac{K(K-1) C_\eps}{2\lambda^2}.$$
The limit now follows after taking $\lambda\to\ii$, $\eps\to0$ and finally $\eta\to0$.
\end{proof}

\begin{remark}[Adiabatic connection]\label{rmk:adiabatic_conn}\rm
For homogeneous potentials such as Coulomb, scaling $\rho$ is the same as changing the strength of the interaction. This is the spirit of the \emph{adiabatic connection formula}, which is often used in quantum chemistry to interpolate between the non-interacting and interacting problems~\cite{HelTea-22}. Let us for instance discuss $F^{w}_{\rm L}[\rho]$ and the corresponding kinetic energy $T_{\rm GC}[\rho]=F^{0}_{\rm L}[\rho]$. We introduce a coupling constant $t$ in front of $w$ and look at the function $t\mapsto F^{tw}_{\rm L}[\rho]$. It is concave on $[0,1]$ (and increasing if $w\geq0$). It has left and right derivatives everywhere, which are given by the minimal and maximal values of the interaction energy among all the possible minimizers $\Gamma_t$ of $F^{tw}_{\rm L}[\rho]$, by the Feynman-Hellmann theorem. These two derivatives are equal, except possibly on a countable subset of $[0,1]$. We can express
$$F^{w}_{\rm L}[\rho]-T_{\rm GC}[\rho]=\int_0^1\frac{\partial}{\partial t}F^{tw}_{\rm L}[\rho]\,\rd t=\int_0^1\tr\bigg(\sum_{1\leq j<k\leq N}w(\br_j-\br_k)\bigg)\Gamma_t\,\rd t$$
where $\Gamma_t$ is any minimizer for $F^{tw}_{\rm L}[\rho]$. This is a formula for the direct plus exchange-correlation energy in Kohn-Sham theory with fractional occupations. It is sometimes useful to consider a general path $t\in[0,1]\mapsto w_t$ in place of the simple linear switching, see~\cite{Yang-98} and~\cite[Sec.~2.4]{HelTea-22}.
\end{remark}

\section{The classical interaction energy and Lieb-Oxford inequalities}\label{sec:classical}

In this section we study the Levy-Lieb functional with the kinetic energy dropped, which then becomes a purely classical problem.

\subsection{A multi-marginal optimal transport problem}
In the classical problem there is no difference between fermions and bosons. In the canonical setting, the main variable is a symmetric probability density $\bP(\br_1,...,\br_N)$ over $(\R^d)^N$ which in the quantum case corresponds to
$$\bP(\br_1,...,\br_N)=\sum_{\sigma_1,...,\sigma_N\in\Z_q}|\Psi(\br_1,\sigma_1,\dots, \br_N,\sigma_N)|^2$$
for pure states and to an average of such quantities for mixed states. The sum over the spin variables occurs since the interaction potential has been assumed to be spin-independent.
In general, $\bP$ will not be absolutely continuous with respect to the Lebesgue measure, however. The problem is therefore better stated in the form
\begin{equation}
\boxed{F_{\rm SCE}[\rho]=\inf_{\substack{\bP\ :\\ \rho_\bP=\rho}}\int_{(\R^d)^N}\sum_{1\leq j<k\leq N}w(\br_j-\br_k)\,\rd\bP(\br_1,...,\br_N) }
\label{eq:def_SCE}
\end{equation}
with the density
$$\rho_\bP(\br)=N\int_{(\R^d)^{N-1}}\,\rd\bP(\br,\br_2,...,\br_N).$$
The acronym SCE means \emph{Strictly Correlated Electrons}~\cite{Seidl-99,SeiPerLev-99,SeiGorSav-07,GorSei-10,SeiMarGerNenGieGor-17,SeiBenKooGor-22} since, as we will explain, the minimizing solution $\bP$ is typically supported on a set of small dimension where the positions of the particles are highly dependent of each other. In general $\rho$ can be a singular measure. In the worst case $\rho$ is the sum of $N$ Dirac deltas, in which case $\bP$ has to be the symmetrized tensor product of these $N$ deltas so that the locations of the particles are then completely fixed. For simplicity we will always assume that $\rho\in L^1(\R^d)$. Nevertheless, the minimizing $\bP$ need not be a function.

There is a grand-canonical version of $F_{\rm SCE}$ which is stated in the form
\begin{equation}
\boxed{F_{\rm GSCE}[\rho]=\inf_{\substack{\bP=(\bP_n)_{n\geq0}\\ \sum_{n\geq0}\bP_n((\R^d)^n)=1\\ \sum_{n\geq1}\rho_{\bP_n}=\rho}}\sum_{n\geq2}\int_{(\R^d)^n}\sum_{1\leq j<k\leq n}w(\br_j-\br_k)\,\rd\bP_n(\br_1,...,\br_n). }
\label{eq:def_SCE_GC}
\end{equation}
This was introduced in~\cite{LewLieSei-18} and further studied in~\cite{LewLieSei-19,LewLieSei-19b,MarLewNen-22_ppt}. For the problem to be well posed for all densities, it is needed that $w$ satisfies the stability condition~\eqref{eq:w_class_stable}.

The two classical problems~\eqref{eq:def_SCE} and~\eqref{eq:def_SCE_GC} belong to the class of \emph{multi-marginal optimal transport} problems~\cite{CotFriKlu-13,CotFriPas-15,Pass-15,MarGerNen-17,SeiMarGerNenGieGor-17}.
We only mention here a few striking results. The existence of a minimizing $\bP$ for~\eqref{eq:def_SCE} follows by compactness arguments similar to Theorem~\ref{thm:exists_min}, for a large class of interaction potentials including the Coulomb potential. The argument is the same in the grand-canonical case~\eqref{eq:def_SCE_GC}.
It was proved in~\cite{ColMar-15} that the infimum can be restricted to \emph{Monge states} which are the most correlated $N$-particle probability densities with one-particle density $\rho$ and take the form
\begin{equation}
\bP(\br_1,...,\br_N)=\text{Sym}\int_{\R^d}\delta_{\by}(\br_1)\delta_{T\by}(\br_2)\cdots \delta_{T^{N-1}\by}(\br_N)\,\frac{\rho(\by)}{N}\,\rd\by,
 \label{eq:Monge}
\end{equation}
where $T:\R^d\to\R^d$ is a transport map such that $T\#\rho=\rho$ and $T^N={\rm Id}$ and Sym denotes  symmetrization. The formula means that the position $\by=\br_1$ of the first particle completely determines the positions $\br_2=T\br_1,...,\br_N=T^{N-1}\br_1$ of the other $N-1$ particles through the transport map $T$ (and the picture is symmetrized with respect to the indices of the particles at the end). When moving the first particle (at the appropriate speed such as to build the desired density $\rho$) the other particles follow in a `strictly correlated' way.

Even if the infimum in~\eqref{eq:def_SCE} is the same when restricted to Monge states, there might exist no Monge minimizer~\cite{ColStr-16,SeiMarGerNenGieGor-17}. Only when $N=2$, or in one dimension for all $N\geq2$~\cite{ColPasMar-15} one can be sure that Monge minimizers exist.
In fact, in dimension $d=1$ and for a positive interaction $w\geq0$, the problem  admits a minimizer $\bP$ which does not depend on $w$ at all! It is the Monge state with increasing transport map $T\#(\rho\1_{(r_{k-1},r_k)})=\rho\1_{(r_{k},r_{k+1})}$  where $r_0=-\ii<r_1<r_2<\cdots <r_{N-1}<r_N=+\ii$ are chosen such that $\int_{r_i}^{r_{i+1}}\rho(r)\,\rd r=1$~\cite{ColPasMar-15}. The corresponding $N$-particle probability can also be expressed in the manner
\begin{equation}
\bP(r_1,...,r_N)=\text{Sym}\int_0^1\delta_{r_1(s)}\otimes\cdots\otimes\delta_{r_N(s)}\,\rd s
\label{eq:form_Monge_1D}
\end{equation}
where $r_k(s):[0,1]\to [r_{k-1},r_k]$ is the inverse of the increasing function $r\mapsto s_k(r)=\int_{r_{k-1}}^r\rho(t)\,\rd t$. This is displayed in Figure~\ref{fig:1D_Monge}.
For instance, for the uniform density $\rho(r)=\1(0\leq r\leq N)$ we have simply $T(y)=y$. The $N$ points are placed on the lattice, $(y+\Z)\cap [0,N)$ and their position is averaged over $s\in[0,1]$:
\begin{equation}
 \bP(r_1,...,r_N)=\text{Sym} \int_{0}^1\delta_{s}(r_1)\delta_{1+s}(r_2)\cdots \delta_{N-1+s}(r_N)\,\rd s.
 \label{eq:floating_crystal_1D}
\end{equation}
This is called a \emph{floating Wigner crystal} in Physics and Chemistry~\cite{BisLuh-82,MikZie-02,DruRadTraTowNee-04,LewLieSei-19b}, since the particles are exactly located on a lattice, whose position is varied. We will come back to this special state later in Section~\ref{sec:floating_crystal}.

\begin{figure}[h]
\centering
\includegraphics[width=9cm]{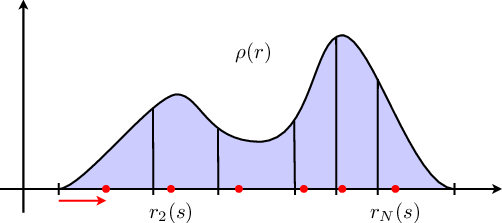}
\caption{Form~\eqref{eq:form_Monge_1D} of the optimal Monge-type probability $\bP$ in one dimension. The positions of all the particles are fixed by the position of the first particle and they are moved to the right at a proper speed such as to reproduce the desired density $\rho$.\label{fig:1D_Monge}}
\end{figure}

In~\cite{ButChaPas-18} some interesting properties of the exact (Monge or not Monge) minimizer $\bP$ of~\eqref{eq:def_SCE} were established. This includes the fact that the $N$ particles have a positive distance to each other on the support of $\bP$, for a repulsive interaction such as Coulomb. The dual formulation is similar to~\eqref{eq:duality} and takes the form
\begin{equation}
 F_{\rm SCE}[\rho]=\sup_{\substack{v\in C^0(\R^d)\\ \sum_{1\leq j<k\leq N}w(\br_j-\br_k)+\sum_{j=1}^Nv(\br_j)\geq0}} \left\{-\int_{\R^d}v(\br)\rho(\br)\,\rd\br\right\}.
 \label{eq:duality_SCE}
\end{equation}
One important feature of the SCE problem is that there exists an optimal potential $v_{\rm SCE}$ solving the supremum in~\eqref{eq:duality_SCE}, under rather weak assumptions on the interaction potential $w$. The optimal $v_{\rm SCE}$ is called a \emph{Kantorovich} potential. This is in stark contrast to the quantum case~\eqref{eq:duality}, where unique continuation drastically reduces the set of densities $\rho$ for which the supremum is attained (see Appendix~\ref{sec:HK} below). Under the sole assumption that $w$ is radial decreasing, diverges at the origin and is $C^1$ outside of the origin (like for the Coulomb potential $w(\br)=|\br|^{-1}$ in dimension $d=3$), it was proved in~\cite{ButChaPas-18} that there exists an optimal Kantorovich potential $v_{\rm SCE}$ which is bounded and Lipschitz. An optimal $N$-particle probability $\bP$ must then be supported on the set
$${\rm argmin}\left\{\sum_{1\leq j<k\leq N}w(\br_j-\br_k)+\sum_{j=1}^Nv_{\rm SCE}(\br_j)\right\}.$$
In other words, the $N$ particles should minimize the associated $N$-particle classical problem with the external potential $v_{\rm SCE}$. All  densities $\rho\in L^1(\R^3)$ are $v$-representable in the classical case.

\subsection{Convergence of the Levy-Lieb functional at low density}

We have seen above in Theorem~\ref{thm:CV_T} that the kinetic energy functional becomes dominant at large densities. Similarly, the interaction becomes dominant at low densities, provided that $w$ has the right scaling at large distances. To simplify our exposition, from now on we restrict our discussion to power-law (Riesz) potentials
$$w(\br)=\frac{1}{|\br|^s},\qquad 0<s<\min(2,d).$$
The main result is the following

\begin{theorem}[Convergence at low density]\label{thm:CV_SCE}
Let $w(\br)=|\br|^{-s}$ with $0<s<\min(2,d)$. Let $\rho\geq0$ such that $\sqrt\rho\in H^1(\R^d)$. If $\int_{\R^d}\rho\in\N$ we have
\begin{equation}
\lim_{\lambda\to0}\frac{F_{\rm LL}\big[\lambda^d\rho(\lambda\cdot)\big]}{\lambda^{1+\frac{s}{d}}}=\lim_{\lambda\to0}\frac{F_{\rm L}\big[\lambda^d\rho(\lambda\cdot)\big]}{\lambda^{1+\frac{s}{d}}}=F_{\rm SCE}[\rho].
\label{eq:limit_F_LL_SCE}
\end{equation}
If $\int_{\R^d}\rho\in\R_+$ we have
\begin{equation}
\lim_{\lambda\to0}\frac{F_{\rm GC}\big[\lambda^d\rho(\lambda\cdot)\big]}{\lambda^{1+\frac{s}{d}}}=F_{\rm GSCE}[\rho].
\label{eq:limit_F_GC_SCE}
\end{equation}
\end{theorem}

Note that the stability condition~\eqref{eq:stability} is always satisfied for the positive potential $w(\br)=|\br|^{-s}$, hence the grand-canonical energies are well defined.

The proof is much more complicated than Theorem~\ref{thm:CV_T}, since an optimizer $\bP$ for $F_{\rm SCE}[\rho]$ or $F_{\rm GSCE}[\rho]$ will never have a finite kinetic energy, even under the assumption that $\sqrt\rho\in H^1(\R^d)$. The limit for $F_{\rm LL}$ was shown for $\int_{\R^d}\rho=2$ with spin first by Cotar, Friesecke and Kl\"uppelberg in~\cite{CotFriKlu-13} and later extended to $\int_{\R^d}\rho=3$ by Bindini and de Pascale in~\cite{BinPas-17}. The limit for $F_{\rm L}[\rho]$ and all $\int_{\R^d}\rho\in\N$ was solved in~\cite{Lewin-18} whereas the case of $F_{\rm LL}$ was finally treated in~\cite{CotFriKlu-18}. The proof for $F_{\rm GC}$ follows along the lines of~\cite{Lewin-18}. The next order in $\lambda$ was predicted in~\cite{GorVigSei-09} and proved in some cases in~\cite{ColMarStr-21_ppt}.

To summarize, at large densities the Levy-Lieb functional behaves like the kinetic energy of non-interacting quantum particles, whereas at low density the particles tend to be very correlated and solve the corresponding classical problem.

\subsection{Lieb-Oxford inequality}

We discuss here upper and lower bounds on the interaction energy, with an emphasis on lower bounds (Lieb-Oxford inequality).

The easiest upper bound is obtained by taking the decorrelated trial state $\bP=(\rho/N)^{\otimes N}$, that is, independent particles distributed according to the density $\rho$. This gives the bound
\begin{equation}
F_{\rm SCE}[\rho]\leq \frac{1-\frac1N}{2}\iint_{\R^{2d}}w(\br-\br')\rho(\br)\rho(\br')\,\rd\br\,\rd\br'.
\label{eq:upper_SCE}
\end{equation}
The right side is, up to the constant $1-1/N$, the classical energy of the density distribution $\rho$ and it is a non-local term. The factor $1/N$ can be dropped for repulsive potentials.

It is relatively easy to prove a similar lower bound, under the additional assumption that $w$ is continuous and has a non-negative Fourier transform, $\widehat{w}\geq0$. In this case we have
$$\iint_{\R^{2d}}w(\br-\br')\,\rd\eta(\br) \,\rd\eta(\br')=(2\pi)^{d/2}\int_{\R^d}\widehat{w}(k)|\widehat{\eta}(k)|^2\,\rd k\geq0$$
for every signed measure $\eta$. Taking $\eta=\sum_{j=1}^N\delta_{\br_j}-f$ and expanding we find the pointwise inequality on $(\R^d)^N$
$$\sum_{1\leq j<k\leq N}w(\br_j-\br_k)\geq \sum_{j=1}^N w\ast f(\br_j)-\frac12\iint_{\R^{2d}}w(\br-\br')f(\br)f(\br')\,\rd\br \,\rd\br'-\frac{w(0)N}{2}.$$
This is valid for all $f$ and the last error term comes from the case $j=k$. Integrating against any state $\bP$ with density $\rho$ and taking $f=\rho$, we obtain the following lower bound:
\begin{equation}
F_{\rm SCE}[\rho]\geq \frac12\iint_{\R^{2d}}w(\br-\br')\rho(\br)\rho(\br')\,\rd\br\,\rd\br'-\frac{w(0)}{2}\int_{\R^d}\rho(\br)\,\rd\br,\quad \text{when $\widehat{w}\geq0$}.
\label{eq:Onsager}
\end{equation}
For a long-range potential the first term  grows faster than $N$ for most densities, hence the last error term is often much lower than the classical energy.

For Coulomb or other power-law potentials, the previous argument does not work since $w(0)=+\ii$. One solution is to regularize the potential at the origin but this also modifies the classical interaction energy. One can estimate the error under appropriate regularity assumptions on $\rho$. But Lieb~\cite{Lieb-79} and then Lieb-Oxford~\cite{LieOxf-80} have proved a universal bound which has the right scaling behavior and does not require to smear out the potential. We state it for power-law potentials but the inequality is slightly more general.

\begin{theorem}[Lieb-Oxford inequality~\cite{Lieb-79,LieOxf-80,LieSei-09,Bach-92,GraSol-94,LieSolYng-95,LunNamPor-16}]
Assume that $w(\br)=|\br|^{-s}$ with $0<s<d$ in dimension $d\geq1$. Then there exists a universal constant $c_{\rm LO}(s,d)>0$ such that
\begin{equation}
 F_{\rm GSCE}[\rho]\geq \frac12\iint_{\R^{2d}}\frac{\rho(\br)\rho(\br')}{|\br-\br'|^s}\,\rd\br\,\rd\br'-c_{\rm LO}(s,d)\int_{\R^d}\rho(\br)^{1+\frac{s}{d}}\,\rd\br,
 \label{eq:LO}
\end{equation}
for every $\rho\in (L^1\cap L^{1+\frac{s}{d}})(\R^d,\R_+)$.
\end{theorem}

From now on we always call $c_{\rm LO}(s,d)$ the \emph{smallest} constant for which the inequality~\eqref{eq:LO} is valid for all $\rho$. Note that $c_{\rm LO}(s,d)$ works for every particle number $\int_{\R^d}\rho$. If one adds the constraint that $\int_{\R^d}\rho=\lambda$ then the optimal constant depends on $\lambda$ but it is non-decreasing and has the limit $c_{\rm LO}(s,d)$ when $\lambda\to\ii$.

Although only the case $s=1$ and $d=3$ was considered in the original papers~\cite{Lieb-79,LieOxf-80}, the proof for $s=1$ and $d=2$ given in~\cite{Bach-92,GraSol-94,LieSolYng-95} extends to any $0<s<d$ in any dimension, see~\cite[Lemma~16]{LunNamPor-16}. This proof involves the Hardy-Littlewood estimate for the maximal function $M_\rho$~\cite{Grafakos-book},
$$\norm{M_\rho}_{L^{1+s/d}(\R^d)}\leq c_{\rm HL}(s,d)\norm{\rho}_{L^{1+s/d}(\R^d)}$$
and, consequently,  the best known estimate on $c_{\rm LO}(s,d)$ involves the unknown constant $c_{\rm HL}(s,d)$. A Lieb-Oxford bound was shown for $w(\br)=-\log|\br|$ in two dimensions in~\cite[Prop.~3.8]{LewNamSerSol-15}. In dimension $d=1$, optimal Lieb-Oxford bounds are studied in~\cite{DiMarino-19}.

In the 3D Coulomb case, $d=3$ and $s=1$, the best estimate known so far on the optimal Lieb-Oxford constant is
\begin{equation}
\boxed{ 1.4442\leq c_{\rm LO}(1,3)\leq 1.5765.}
 \label{eq:estim_c_LO}
\end{equation}
The upper constant was equal to $8.52$ in~\cite{Lieb-79}, to $1.68$ in~\cite{LieOxf-80} and later improved to 1.64 in~\cite{ChaHan-99}.
The better value $1.58$ was obtained very recently in~\cite{LewLieSei-22}. The lower bound has been claimed in~\cite{Perdew-91,LevPer-93} and only shown recently in~\cite{CotPet-19b,LewLieSei-19}. It will be discussed in the next section. It was  conjectured in~\cite{LevPer-93,OdaCap-07,RasPitCapPro-09} that the best Lieb-Oxford constant is indeed about $1.44$. It remains an important challenge to find the optimal constant in~\eqref{eq:LO}. Several of the most prominent functionals used in Density Functional Theory make use of the value of the Lieb-Oxford constant for calibration~\cite{Perdew-91,LevPer-93,PerBurErn-96,SunPerRuz-15,SunRuzPer-15,Perdew_etal-16,PerSun-22}.

A different Lieb-Oxford inequality was recently proved in the 3D Coulomb case in~\cite{LewLieSei-22}. It reads
\begin{multline}
\iint_{\R^{3N}}\left(\sum_{1\leq j<k\leq N}\frac{1}{|\br_j-\br_k|}\right)\,\rd\bP(\br_1,...,\br_N)- \frac12\iint_{\R^{6}}\frac{\rho_\bP(\br)\rho_\bP(\br')}{|\br-\br'|}\,\rd\br\,\rd\br'\\ \geq-1.2490\int_{\R^3}\rho_\bP(\br)^{\frac43}\,\rd\br,
\label{eq:exchange}
\end{multline}
under the additional assumption that $\bP$ has \emph{negative correlations}, which means
\begin{equation}
 N(N-1)\iint_{\R^{3N-6}}\rd\bP(\br,\br',\br_3,...,\br_N)\leq \rho_\bP(\br)\,\rho_\bP(\br'),\qquad \text{for a.e. $\br,\br'\in\R^3$.}
 \label{eq:negative_corr}
\end{equation}
This condition is satisfied when $\bP$ is the square of a Slater determinant~\eqref{eq:Slater_det}, in which case the left side of~\eqref{eq:exchange} is called the \emph{exchange energy}  and the best constant is believed to be $1.09$~\cite{PerSun-22}. But many other states satisfy the condition~\eqref{eq:negative_corr}. In statistical mechanics, this is typical of gas phases~\cite{Ruelle} at high temperature. Since $1.25<1.44<c_{\rm LO}(1,3)$, this means that such states cannot provide the optimal Lieb-Oxford constant. In fact, we explain below how to obtain the constant $1.44$ from a solid (periodic) phase.

The \emph{indirect energy} is the equivalent of the exchange-correlation energy defined in the quantum case in Section~\ref{sec:KS}:
$$E_{\rm Ind}[\rho]=F_{\rm SCE}[\rho]- \frac12\iint_{\R^{2d}}w(\br-\br')\rho(\br)\rho(\br')\,\rd\br\,\rd\br'.$$
For power-law interactions it is always negative and bounded from below by a constant times $\int_{\R^d}\rho^{1+s/d}$.

\subsection{Constant densities and the classical Uniform Electron Gas}\label{sec:floating_crystal}
We discuss here the special case of densities which are constant over a finite set and the limit when this set fills the whole space. This is the classical equivalent of the \emph{Uniform Electron Gas} discussed in the quantum case in Section~\ref{sec:UEG} above. This special case will give us some lower bounds on the Lieb-Oxford constant $c_{\rm LO}(s,d)$, including the bound 1.44 in dimension $d=3$ stated in~\eqref{eq:estim_c_LO}.

The classical equivalent of Theorem~\ref{thm:UEG} was  proved in~\cite{LewLieSei-18}.

\begin{theorem}[The classical Uniform Electron Gas energy~\cite{LewLieSei-18}]\label{thm:UEG_classical}
Assume that $w(\br)=|\br|^{-s}$ with $0<s<d$ in dimension $d\geq1$. Let $\rho_0>0$. Let $\Omega$ be a fixed open convex set of unit volume $|\Omega|=1$.
Then there exists a universal constant $c_{\rm UEG}(s,d)>0$ such that
\begin{align}
&\lim_{L\to\ii}L^{-d}\left( F_{\rm GSCE}[\rho_0\1_{L\Omega}]-\frac{\rho_0^2}2\iint_{(L\Omega)^2}\frac{\rd\br\,\rd\br'}{|\br-\br'|^s}\right)\nn\\
&\qquad\qquad=\lim_{\substack{L\to\ii\\ L^d\in\N/\rho_0}}L^{-d}\left( F_{\rm SCE}[\rho_0\1_{L\Omega}]-\frac{\rho_0^2}2\iint_{(L\Omega)^2}\frac{\rd\br\,\rd\br'}{|\br-\br'|^s}\right)\nn\\
&\qquad\qquad=c_{\rm UEG}(s,d)\rho_0^{1+\frac{s}{d}}
 \label{eq:UEG_classical}
\end{align}
In particular we obtain $c_{\rm LO}(s,d)\geq -c_{\rm UEG}(s,d)$.
\end{theorem}

The constant $c_{\rm UEG}(1,3)$ is the one which has appeared before in Theorem~\ref{thm:prop_f}. At low density, the quantum UEG behaves like a classical gas by an equivalent of Theorem~\ref{thm:CV_SCE} for infinite systems~\cite{LewLieSei-18}.

Note that the classical canonical and grand-canonical functionals are known to give the same thermodynamic limit. In the quantum case this is not yet known. We have stated the theorem for a fixed domain $\Omega$ which is scaled but the same result holds for a general sequence $\Omega_L$ that has a regular boundary in the sense of Fisher~\cite{LewLieSei-18}.

Except in dimensions $d\in\{1,8,24\}$ and $\max(0,d-2)\leq s<d$, some special cases to which we will come back, the constant $c_{\rm UEG}(s,d)$ is unknown. In order to get upper bounds on $c_{\rm UEG}(s,d)$, we need to construct trial states. The idea is to use a \emph{floating crystal} similar to~\eqref{eq:floating_crystal_1D}, that is, to place the particles on a lattice and then average over translations to obtain a constant density.

Let $\mathscr{L}\subset\R^d$ be a lattice of normalized unit cell $Q$. We then only retain the points of the lattice intersecting the large cube $C_L=(-L/2,L/2)^d$ and average over the translations of this finite lattice over $Q$. This way we obtain a trial state which is constant over the union of the corresponding translates of $Q$. In general this is only an approximation of $C_L$ but since the limit~\eqref{eq:UEG_classical} is  insensitive to the type of domains, this will not create any difficulty. The trial state is, therefore, given by
\begin{equation}
 \bP_{\mathscr{L},L}:=\text{Sym}\int_{Q}\bigotimes_{\ell\in\mathscr{L}\cap C_L}\delta_{\ell +\by}\,\rd\by.
 \label{eq:floating_Wigner}
\end{equation}
Then $\bP_{\mathscr{L},L}$ has the constant density
\begin{equation}
 \rho_{\bP_{\mathscr{L},L}}=\1_{\Omega_L}\qquad\text{over the set }\qquad \Omega_L=\bigcup_{\ell\in\mathscr{L}\cap C_L}Q+\ell.
 \label{eq:Omega_L}
\end{equation}
The state is as displayed in Figure~\ref{fig:floating}.
Note that the energy of the probability measure $\bP_{\mathscr{L},L}$ is simply the interaction of the lattice points, since the interaction potential is translation-invariant:
$$\int_{(\R^d)^N}\sum_{1\leq j<k\leq N}\frac{1}{|\br_j-\br_k|^s}\,\rd\bP_{\mathscr{L},L}=\frac12 \sum_{\substack{\ell\neq\ell'\\ \in \mathscr{L}\cap C_L}}\frac{1}{|\ell-\ell'|^s}.$$

It is instructive to see first what happens in the short range case $s>d$. Then the energy per unit volume converges to
\begin{equation}
\lim_{L\to\ii}\frac1{2|\Omega_L|} \sum_{\substack{\ell\neq\ell'\\ \in \mathscr{L}\cap C_L}}\frac{1}{|\ell-\ell'|^s}=\frac12\sum_{\ell\in\mathscr{L}\setminus\{0\}}\frac{1}{|\ell|^s}=:\zeta_{\mathscr{L}}(s),\qquad s>d.
 \label{eq:Epstein}
\end{equation}
The function on the right side is called the \emph{Epstein Zeta function}~\cite{Epstein-03,BorGlaMPh-13,BlaLew-15} and it is the natural generalization to $\R^d$ of the usual Riemann Zeta function, with which it coincides when $d=1$ (hence $\mathscr{L}=\Z$). It turns out that the limit in the long range case can be expressed with the (analytic extension) of $\zeta_{\mathscr{L}}$, for potentials decaying to zero at infinity faster than Coulomb. Something special is happening at $s=d-2$.

\begin{theorem}[Indirect energy of the floating Wigner crystal~{\cite{BorBorShaZuc-88,BorBorSha-89,BorBorStr-14,LewLie-15,Lauritsen-21,Lewin-22}}]\label{thm:shift}
Let $d-2\leq s<d$ in dimension $d\geq3$ and $0<s<d$ in dimensions $d=1,2$. Let $\mathscr{L}\subset\R^d$ be a lattice with a normalized unit cell $Q$ having no dipole and no quadrupole moment:
$$\int_{Q}\br\,\rd\br=0,\qquad \int_{Q}r_ir_j\,\rd\br=\frac{\delta_{ij}}{d}\int_Q|\br|^2\,\rd\br.$$
Then the indirect energy per unit volume of the floating Wigner crystal~\eqref{eq:floating_Wigner} converges to
\begin{multline}
\lim_{L\to\ii}|\Omega_L|^{-1}\left(\dps\frac12 \sum_{\substack{\ell\neq\ell'\\\in \mathscr{L}\cap C_L}}\frac{1}{|\ell-\ell'|^s}-\frac12\iint_{(\Omega_L)^2}\frac{\rd\br\,\rd\br'}{|\br-\br'|^s}\right)\\
=\begin{cases}
\zeta_\mathscr{L}(s)&\text{for $s>d-2$,}\\
\zeta_\mathscr{L}(d-2)+\dps\frac{|\mathbb{S}^{d-1}|}{2d}\int_Q|\br|^2\,\rd\br&\text{for $s=d-2$,}\\
\end{cases}
\label{eq:computation_floating}
\end{multline}
where $\zeta_\mathscr{L}(s)$ is the analytic continuation to $\C\setminus\{d\}$ of the Epstein Zeta function on the right of~\eqref{eq:Epstein}, initially defined for $\Re(s)>d$.
\end{theorem}

A similar result holds for $-1\leq s\leq 0$ in $d=1$ and $s=0$ in $d=2$~\cite{Lauritsen-21,Lewin-22}.

The first divergent term in the lattice sum is the classical energy, which behaves like $N^{2-s/d}$ and depends on the shape of the chosen large domain $C_L$:
\begin{equation}
\frac12\iint_{(\Omega_L)^2}\frac{\rd\br\,\rd\br'}{|\br-\br'|^s}\sim_{L\to\ii}\frac{L^{2d-s}}{2}\iint_{(C_1)^2}\frac{\rd\br\,\rd\br'}{|\br-\br'|^s}.
\label{eq:main_term_lattice_sum}
\end{equation}
This is because the lattice sum is a Riemann sum for the corresponding integral at that scale. Replacing $C_L$ by another set changes this macroscopic term. Note that the analytic extension of~\eqref{eq:main_term_lattice_sum} is a $o(L^d)$ for $s>d$. This term probably exists in the short range case too, but it is lower order and it was not seen in the limit~\eqref{eq:Epstein}.

The theorem provides the next order term in the long range case $d-2 < s<d$. This is an extensive quantity (of the order of the volume) which has a limit independent of the shape $C_L$. This limit is simply the analytic extension of the short range energy. This is compatible with our picture that the classical energy~\eqref{eq:main_term_lattice_sum} is the leading term for $s<d$ but once it is removed, we are essentially back to~\eqref{eq:Epstein}.

At $s=d-2$ the picture changes. Another term of the order $L^{2(d-1)-s}$, which was lower order for all $s>d-2$, becomes relevant for the energy per unit volume at $s=d-2$ and dominates for $s<d-2$. As we will explain later, this is a kind of surface term.

In dimension $d=1$ we know from~\cite{ColPasMar-15} that the floating crystal is optimal and provides the minimal classical energy at constant density. Therefore we deduce from Theorem~\ref{thm:shift} that
\begin{equation}
c_{\rm UEG}(s,1)=\zeta(s),\qquad \text{for all $0<s<1$.}
\end{equation}
In particular, $-c_{\rm LO}(s,1)\leq\zeta(s)$.

In higher dimensions, the floating crystal is not known to be an exact minimum, and furthermore there are several possible crystals. Therefore  we only obtain the upper bound
\begin{equation}
 -c_{\rm LO}(s,d)\leq c_{\rm UEG}(s,d)\leq \min_{\mathscr{L}}\zeta_\mathscr{L}(s),\qquad \text{for }\max(0,d-2)<s<d.
 \label{eq:bound_LO}
\end{equation}
The minimum is over all lattices of normalized unit cell. It is expected that the last inequality should be an equality for some values of the dimension $d$ including $d=1,2,3$. So far this is only known in dimensions $d=8$ and $d=24$~\cite{CohKumMilRadVia-19_ppt,PetSer-20}. In dimension $d=2$ the minimum on the right of~\eqref{eq:bound_LO} is known to be achieved by the triangular lattice~\cite{Rankin-53,Cassels-59,Ennola-64,Diananda-64,Montgomery-88} whereas in dimension $d=3$, numerics indicates that it is achieved by the Body-Centered Cubic lattice (BCC) for $0<s\leq3/2$ and the Face-Centered Cubic lattice (FCC) for $3/2\leq s<3$~\cite{GiuVig-05,SarStr-06,BorGlaMPh-13,BlaLew-15}.

The surprising jump of the energy per unit volume~\eqref{eq:computation_floating} in the Coulomb case $s=d-2$ was first discovered in 1979 by Hall~\cite{Hall-79} based on an unpublished remark by Plaskett in 1959. The conendrum raised by Hall was discussed in several papers in the 80s, see for instance~\cite{DeWette-80,IhmCoh-80,HalRic-80,Hall-81,AlaJan-81,NijRui-88}. It was rediscovered in~1988 by Borwein \emph{et al}~\cite{BorBorShaZuc-88} and was recently revived and reformulated in~\cite[App.~B]{LewLie-15}. It has indeed always been assumed in the Physics and Chemistry literature that the floating crystal is a good trial state for the UEG, and that $c_{\rm UEG}(1,3)$ should even be equal to the BCC lattice energy, whose value is $\zeta_{\rm BCC}(1)\simeq -1.4442$ (see~\cite{ColMar-60} and \cite[p.~43]{GiuVig-05}). This value is used in most DFT functionals based on the Uniform Electron Gas. But the Coulomb potential is exactly the one for which the floating crystal behaves badly, by Theorem~\ref{thm:shift}.

Note that the jump exists and is unavoidable in 1D, where the floating crystal is known to be optimal, but it happens at the negative value $s=-1$. Indeed, for $w(r)=-|r|$ we have the expansion similar to Theorem~\ref{thm:shift}
\begin{equation}
F_{\rm SCE}\big[\rho_0\1_{[0,L]}\big]=-\frac{L^3(\rho_0)^2}2\int_0^1\int_0^1|x-y|\,\rd x\,\rd y+\frac{L(\rho_0)^2}6+O(1)
\label{eq:UEG_1D}
\end{equation}
with $1/6>-\zeta(-1)=1/12$.

In order to better understand what is going on, it is useful to reinterpret the result in terms of the \emph{Jellium model}~\cite{LieNar-75,Lewin-22}. In this model there is no constraint on the electronic density but the particles interact with a compensating uniform background of opposite charge. At density one, the corresponding energy is defined by
\begin{equation*}
\cE_{\rm Jel}(\Omega,\br_1,...,\br_N)=\sum_{1\leq j<k\leq N}\frac{1}{|\br_j-\br_k|^s}-\sum_{j=1}^N\int_{\Omega}\frac{\rd\br'}{|\br_j-\br'|^s}+\frac12\iint_{\Omega\times\Omega}\frac{\rd\br\,\rd\br'}{|\br-\br'|^s}
\end{equation*}
where $\Omega$ is any measurable set of volume $|\Omega|=N$, representing the uniform background.
A short calculation  shows that the indirect energy of the floating crystal can be written in the form
\begin{equation}
 \frac12 \sum_{\substack{\ell\neq\ell'\\ \in \mathscr{L}\cap C_L}}\frac{1}{|\ell-\ell'|^s}-\frac12\iint_{(\Omega_L)^2}\frac{\rd\br\,\rd\br'}{|\br-\br'|^s}=\int_{Q}\cE_{\rm Jel}\big(\Omega_L,(\mathscr{L}\cap C_L)+\by\big)\,\rd\by.
\label{eq:formula_Jellium_floating}
 \end{equation}
In other words it is the average of Jellium energies where the lattice points are moved over the fixed background $\Omega_L$. In this interpretation it becomes clear why the averaging over $\by$ induces the shift in~\eqref{eq:computation_floating}: moving the particles away from the center of the unit cells is not at all energetically favorable. When the particles are moved in one direction this creates a large excess of negative charges on one side and a corresponding excess of background charge on the opposite side (Figure~\ref{fig:floating}). These two opposite boundary charges have an interaction energy proportional to $(L^{d-1})^2/L^s=L^{2(d-1)-s}$ which is exactly of the order of the volume in the Coulomb case $s=d-2$ and grows faster for $s<d-2$. On the other hand, when the particles are placed exactly at the center of the unit cells, one recovers the analytic extension of the short range energy for all $s>d-4$.

\begin{figure}[t]
 \includegraphics[width=7cm]{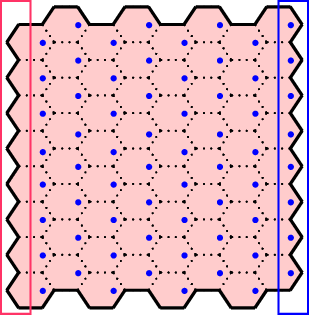}
\caption{A two-dimensional picture of the Jellium model. The dots represent the point particles, which are placed on a finite subset of a lattice $\mathscr{L}$. The colored set is the union of the corresponding unit cells and it represents a uniform background charge distribution of opposite charge. The indirect energy of the floating crystal is obtained after integrating the position of the lattice over the unit cell $Q$ as in~\eqref{eq:formula_Jellium_floating}. When the lattice is not centered, this results in an excess of point charges on one side and an excess of background charge on the other side, indicated by the two rectangles. These large boundary charge fluctuations are responsible for the shift in Theorem~\ref{thm:shift} at $s=d-2$.\label{fig:floating}}
\end{figure}

\begin{theorem}[Jellium energy of the clamped Wigner crystal~{\cite{LieNar-75,BorBorShaZuc-88,BorBorSha-89,BorBorStr-14,LewLie-15,BlaLew-15,Lauritsen-21,Lewin-22}}]\label{thm:no_shift_Jellium}
Assume that $d\geq1$. Let $\mathscr{L}\subset\R^d$ be a lattice satisfying the same assumptions as in Theorem~\ref{thm:shift}. Then the Jellium energy per unit volume of the Wigner crystal clamped at the center of the unit cells, converges to
\begin{equation}
\lim_{L\to\ii}\frac{\cE_{\rm Jel}\big(\Omega_L,\mathscr{L}\cap C_L\big)}{|\Omega_L|}=\zeta_\mathscr{L}(s)
\label{eq:computation_Jellium}
\end{equation}
for all $\max(0,d-4)<s<d$.
\end{theorem}

The floating crystal is really not a good trial state for the UEG. The conundrum raised in Theorem~\ref{thm:shift} was recently resolved in~\cite{CotPet-19b,LewLieSei-19b}.
Cotar and Petrache managed to prove in~\cite{CotPet-19b} that the (unknown) UEG energy $c_{\rm UEG}(s,d)$ is always continuous for $0<s<d$ and that it is equal to the (also unknown) Jellium energy for $d-2\leq s<d$. The proof of continuity in $s$ is very delicate and requires the use of advanced analytical techniques due to Fefferman and collaborators~\cite{Fefferman-85,Hugues-85,Gregg-89}. A short time later, the same result was obtained in~\cite{LewLieSei-19b} with a  different and much simpler argument. Here we only explain this argument for the special case of the floating crystal, that is, we show how to modify the trial state~\eqref{eq:floating_Wigner}  in order to cancel the shift appearing in~\eqref{eq:computation_floating}.

The main idea of~\cite{LewLieSei-19b} is to immerse the crystal in a thin layer of fluid. In other words, the floating crystal is melted close to the boundary in order to reduce the large charge fluctuations. The fluid gets displaced with the crystal when the latter is averaged over translations. Think of a block of ice filling completely a container. In order to move the ice it is necessary to melt it close to the container walls.

To describe this procedure, let us denote by $C'_L$ a slightly larger cubic container such that $\Omega_L+Q\subset C'_L$, where we recall that $\Omega_L$ is the union of the unit cells $Q+\ell$ with $\ell\in \mathscr{L}\cap C_L$. We can take $C'_L=C_{L+\lambda}$ where $\lambda$ is any fixed distance larger than the diameter of $Q$. We assume that the volume of the fluid $|C'_L\setminus \Omega_L|=M$ is an integer. It satisfies $M\leq CL^{d-1}\ll L^d$. The new trial state has the $N=|\Omega_L|$ particles on the floating crystal, translated by $\by\in Q$ as before in~\eqref{eq:floating_Wigner}, together with $M$ other particles forming an uncorrelated fluid in $C'_L\setminus(\Omega_L+\by)$, the set remaining after we have subtracted the union of all the cells centered at the particle positions (Figure~\ref{fig:floating2}):
\begin{equation}
 \tilde\bP_{\mathscr{L},L}=\text{Sym}\int_Q \bigotimes_{\ell\in\mathscr{L}\cap C_L}\delta_{\ell+\by}\otimes\left(\frac{\1_{C'_L\setminus (\Omega_L+\by)}}{M}\right)^{\otimes M}\,\rd\by.
 \label{eq:floating_crystal_modified}
\end{equation}
Note that the state of the fluid is correlated with the position $\by$ of the crystal.

\begin{figure}[t]
 \includegraphics[width=7cm]{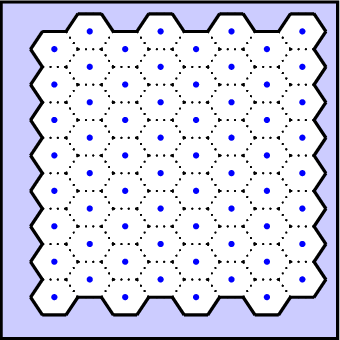}
 \caption{A two-dimensional picture of the modified floating crystal~\eqref{eq:floating_crystal_modified} from~\cite{LewLieSei-19b}. The dots represent the point particles which are at the centers of hexagons of volume one. As the whole crystal block is translated by $\by$, the incompressible fluid gets displaced to fill the remaining space $C'_L\setminus (\Omega_L+\by)$. The resulting density is only constant well inside the container. \label{fig:floating2}}
\end{figure}

\begin{theorem}[Indirect energy of the modified floating crystal~{\cite{LewLieSei-19b}}]
Let $\max(0,d-4)< s<d$ in dimension $d\geq1$ and $\mathscr{L}\subset\R^d$ a lattice satisfying the same assumptions as in Theorem~\ref{thm:shift}.
Then the indirect energy per unit volume of the modified floating Wigner crystal~\eqref{eq:floating_crystal_modified} converges to
\begin{multline}
\lim_{L\to\ii}|C'_L|^{-1}\bigg(\dps\int_{(\R^d)^N}\sum_{1\leq j<k\leq N}\frac{1}{|\br_j-\br_k|^s}\rd\tilde\bP_{\mathscr{L},L}\\
-\frac12\iint_{\R^{2d}}\frac{\rho_{\tilde\bP_{\mathscr{L},L}}(\br)\rho_{\tilde\bP_{\mathscr{L},L}}(\br')}{|\br-\br'|^s}\,\rd\br\,\rd\br'\bigg)
=\zeta_\mathscr{L}(s).
\label{eq:computation_floating_modified}
\end{multline}
In particular we obtain
$$-c_{\rm LO}(s,d)\leq c_{\rm UEG}(s,d)\leq \zeta_\mathscr{L}(s),\qquad\text{ for all $\max(0,d-4)< s<d$.}$$
\end{theorem}

For the BCC lattice in dimension $d=3$, one finds the claimed lower bound
$$c_{\rm LO}(1,3)\geq -\zeta_{\rm BCC}(1)\simeq 1.4442.$$
It is reasonable to conjecture that
$$e_{\rm UEG}(s,d)=\min_{\mathscr{L}}\zeta_{\mathscr{L}}(s)$$
for $d=1,2,3$ and all $0<s<d$, which amounts to saying that the uniform electron gas is always crystallized at zero temperature~\cite{BlaLew-15,Lewin-22}. The modified trial state~\eqref{eq:floating_crystal_modified} suggests that the system can only be a solid in the bulk. In a neighborhood of the boundary, it is probably a fluid because the particles have to be able to move sufficiently far away around the crystal to compensate the large charge fluctuations.

In this section we have explained some major difficulties encountered when trying to construct good trial states for the Uniform Electron Gas. Those are entirely due to the boundary, namely to the fact that we work with a finite piece of material in the physical space $\R^d$. If we set up the model on the torus, as is often done in practical calculations, these difficulties disappear~\cite[Sec.~IV.C]{Lewin-22}.

It is worth mentioning an estimate due to Lieb and Narnhofer~\cite{LieNar-75} in the Coulomb case $s=1$ in dimension $d=3$ which states that
\begin{equation}
 F_{\rm SCE}[\1_{\Omega}]\geq \frac12\iint_{\Omega\times\Omega}\frac{\rd\br\,\rd\br'}{|\br-\br'|}-\frac35 \left(\frac{9\pi}{2}\right)^{\frac13}|\Omega|
 \label{eq:LieNar}
\end{equation}
for \emph{any} open set $\Omega$ of integer volume. The constant $(3/5) ({9\pi}/{2})^{1/3}\simeq 1.4508$ is surprisingly close to the expected optimal value $-\zeta_{\rm BCC}(1)$ and it implies in any case that
$$\boxed{-1.4508\leq c_{\rm UEG}(1,3)\leq -1.4442.}$$
This is the bound that appeared in Theorem~\ref{thm:prop_f}. For negatively-correlated states as in~\eqref{eq:negative_corr}, the constant can be replaced by $(3/2) ({\pi}/{6})^{1/3}\simeq 1.2090$~\cite{LewLieSei-22}.

In dimension $d=3$ for $s=1$ it was conjectured in~\cite{OdaCap-07,RasPitCapPro-09} that the classical Uniform Electron Gas gives the optimal Lieb-Oxford constant, i.e., $c_{\rm LO}(1,3)=-c_{\rm UEG}(1,3)$.

\begin{remark}[Determinantal processes and the Dirac constant]\rm
Let $\rho=\rho_0\1_{C_L}$ with $N=\rho_0 L^d\in\N$. Instead of the floating crystal~\eqref{eq:floating_Wigner} and its modified version~\eqref{eq:floating_crystal_modified}, one can consider the square of a Slater determinant as a trial state:
$$\bP(\br_1,...,\br_N)=\frac{1}{\sqrt{N!}}L^{-Nd}|\det\left(e^{i\frac{2\pi}{L}\bk_i\cdot \br_j}\1_{C_L}(\br_j)\right)|^2,$$
where $\bk_1,...,\bk_N$ are $N$ distinct points in $\Z^d$ (our trial state contains no spin). We find
\begin{multline*}
F_{\rm SCE}[\rho_0\1_{C_L}]\leq \frac{\rho_0^2}2\iint_{(C_L)^2}w(\br-\br')\,\rd\br\,\rd\br'\\
-\frac1{2}\iint_{(C_L)^2}w(\br-\br')\left|L^{-d}\sum_{j=1}^Ne^{i\frac{2\pi}{L}\bk_j\cdot (\br-\br')}\right|^2\,\rd\br\,\rd\br'.
\end{multline*}
As for the free Fermi gas we choose all the points $\bk_i$ in a given ball centered at the origin.
For $w(\br)=|\br|^{-s}$ the second term behaves like $L^d\rho_0^{1+s/d}c_{\rm D}(s,d)$ where the \emph{Dirac constant}~\cite{Dirac-28b}
$$c_{\rm D}(s,d)=\frac1{2(2\pi)^d}\int_{\R^d}\frac{|\widehat{\1_{B_{k_ {\rm F}}}}(\br)|^2}{|\br|^s}\,\rd\br$$
is the exchange energy per unit volume of a free Fermi gas. Here the Fermi radius is $k_ {\rm F}=\sqrt{2(d+2)c_{\rm TF}(d)/d}$. This proves that
$$-c_{\rm LO}(s,d)\leq c_{\rm UEG}(s,d)\leq -c_{\rm D}(s,d)$$
but this bound is worse than the floating crystal. The particles are not correlated enough. In dimension $d=3$ with $s=1$ one finds $c_{\rm D}(1,3)=(3/4)\left(6/\pi\right)^{1/3}\simeq0.9305$. Recall that for an arbitrary determinantal point process, we have the better Lieb-Oxford inequality~\eqref{eq:exchange} from~\cite{LewLieSei-22}.
\end{remark}

\subsection{Local Density Approximation for the classical interaction energy}

We have discussed in the previous section the case of exactly constant densities and their limit of infinite volume. We now consider the case of slowly varying densities, which are assumed to be essentially constant over large sets.

First, we mention that the Lieb-Narnhofer bound~\eqref{eq:LieNar} from~\cite{LieNar-75} was generalized to arbitrary densities, in the form of lower bounds involving gradient-type corrections~\cite{BenBleLos-12,LewLie-15}. For instance, the bound
\begin{multline}
F_{\rm GSCE}[\rho]\geq \frac12\iint_{\R^{6}}\frac{\rho(\br)\rho(\br')}{|\br-\br'|}\,\rd\br\,\rd\br'
-\left(\frac35 \left(\frac{9\pi}{2}\right)^{\frac13}+\eps\right)\int_{\R^3}\rho(\br)^{\frac43}\,\rd\br\\-\frac{0.001206}{\eps^3}\int_{\R^3}|\nabla\rho(\br)|\,\rd\br
\end{multline}
was shown to hold in~\cite{LewLie-15} for any $\varepsilon>0$.

The following result gives a quantitative estimate on the grand-canonical classical energy for slowly varying densities.

\begin{theorem}[Local Density Approximation of the classical Coulomb energy{~\cite{LewLieSei-19}}]
Consider the case $w(\br)=|\br|^{-1}$ in dimension $d=3$. There exists a constant $C$ such that
\begin{multline}
 \left|F_{\rm GSCE}(\rho)-\frac12\iint_{\R^3\times\R^3}\frac{\rho(\br)\rho(\br')}{|\br-\br'|}\,\rd\br\,\rd\br'- c_{\rm UEG}(1,3)\int_{\R^3}\rho(\br)^{\frac43}\,\rd\br\right|\\
 \leq \eps \int_{\R^3}\big(\rho(\br)+\rho(\br)^{\frac43}\big)\,\rd\br
+\frac{C}{\eps^{7}}\int_{\R^3}|\nabla\rho^{\frac13}(\br)|^4\,\rd\br
 \label{eq:LDA_main_estim_classical}
\end{multline}
for every $\eps>0$ and every non-negative density $\rho\in L^1(\R^3)\cap L^{4/3}(\R^3)$ such that $\nabla \rho^{1/3}\in L^4(\R^3)$.
\end{theorem}

The gradient term can be replaced by $\eps^{-b}\int_{\R^3}|\nabla\rho^\theta(\br)|^p\,\rd\br$ for any $p>3$ and $0<\theta<1$ such that $\theta p\geq4/3$, with $b=\max\{2p-1,(1+3\theta)p-4\}$.

If we fix a density $\rho$ with $\int_{\R^3}\rho=1$ and take $\rho_N(\br)=\rho(\br N^{-1/3})$, then we find that
\begin{multline}
F_{\rm GSCE}\big[\rho(\cdot N^{-1/3})\big]=\frac{N^{\frac53}}2\iint_{\R^3\times\R^3}\frac{\rho(\br)\rho(\br')}{|\br-\br'|}\,\rd\br\,\rd\br'\\+c_{\rm UEG}(1,3)N\int_{\R^3}\rho(\br)^{\frac43}\,\rd\br+O\left(N^{\frac56}\right).
\label{eq:expansion_N_classical}
\end{multline}
Compare this expansion with the quantum case~\eqref{eq:LDA_main_rescaled_density}.

It is an open problem to prove an estimate similar to~\eqref{eq:LDA_main_estim_classical} for the canonical SCE functional $F_{\rm SCE}$. However, a non quantitative convergence similar to~\eqref{eq:expansion_N_classical} is known even for $w(\br)=|\br|^{-s}$ in all dimensions $d>s>0$:
\begin{multline}
F_{\rm SCE}\big[\rho(\cdot N^{-1/d})\big]=\frac{N^{2-\frac{s}d}}2\iint_{\R^d\times\R^d}\frac{\rho(\br)\rho(\br')}{|\br-\br'|^s}\,\rd\br\,\rd\br'\\+c_{\rm UEG}(s,d)N\int_{\R^3}\rho(\br)^{1+\frac{s}d}\,\rd\br+o(N).
\label{eq:expansion_N_classical2}
\end{multline}
In the Coulomb case for $d=3$,~\eqref{eq:expansion_N_classical2} was  proved in~\cite{LewLieSei-18} whereas general power-law potentials were covered in~\cite{CotPet-19}.

\section{Upper and lower bounds on the Levy-Lieb functionals}\label{sec:TFD_is_a_lower_bound}

In the previous sections we have studied the kinetic and interaction energies separately and reviewed several known upper and lower bounds. Since the minimum of a sum is always larger or equal to the sum of the minima, we easily obtain lower bounds on the full Levy-Lieb functional. For instance, putting together the Lieb-Thirring and Lieb-Oxford inequalities, we find
\begin{multline*}
F_{\rm GC}[\rho]\geq c_{\rm LT}(d) q^{-\frac2d}\int_{\R^d}\rho(\br)^{1+\frac2d}\,\rd\br+\frac12\iint_{\R^d\times\R^d}\frac{\rho(\br)\rho(\br')}{|\br-\br'|^s}\,\rd\br\,\rd\br'\\
-c_{\rm LO}(s,d)\int_{\R^3}\rho(\br)^{1+\frac{s}d}\,\rd\br,
\end{multline*}
for the interaction $w(\br)=|\br|^{-s}$ in dimension $d>s>0$.
The right side takes the same form as the \emph{Thomas-Fermi-Dirac functional}, except for the value the two constants in front of the terms $\rho^{1+2/d}$ and $\rho^{1+s/d}$. One would obtain the right constant in front of the $\rho^{1+2/d}$ in dimension $d\geq3$ if the Lieb-Thirring conjecture mentioned in Section~\ref{sec:LT} had been proved. We refer to~\cite{Lieb-81b}
for a review of results on Thomas-Fermi-type functionals. From Nam's bound~\eqref{eq:Lieb-Thirring_Nam} one can replace the Lieb-Thirring constant by $(1-\eps)c_{\rm TF}(d)$ at the expense of a (negative) gradient correction.

Upper bounds are more complicated because a trial state that works for $T[\rho]$ could be very bad for the classical energy and conversely. We have seen in Section~\ref{sec:kinetic} that the set of one-particle density matrices that are $N$-representable by a mixed states is exactly given by the operators $\gamma=\gamma^*$ such that $0\leq\gamma\leq1$ and $\tr(\gamma)=N$. This is because any such operator is the convex combination of rank-$N$ projections which correspond to Slater determinants. But for estimating the interaction energy we need some more information on the two-particle density.

An explicit convex combination which provided a bound on the two-particle density matrix was derived in~\cite{Lieb-81}. The idea is the following. Assume for simplicity that $\gamma=\sum_{i=1}^Kn_i|u_i\rangle\langle u_i|$ has finite rank $K$. By Horn's lemma~\cite{Lieb-81} there exists a set of $N$ orthonormal vectors $V^1,...,V^N$ in $\C^K$ such that $\sum_{k=1}^N|V^k_i|^2=n_i$ for all $i$. Define then the new orbitals $f^k_\theta=\sum_{j=1}^Ke^{i\theta_j}V^k_ju_j$ where $\theta=(\theta_1,...,\theta_K)\in (0,2\pi)^K$ and the associated trial mixed state
$$\Gamma=\frac{1}{(2\pi)^K}\int_0^{2\pi}\rd\theta_1\cdots \int_0^{2\pi}\rd\theta_K\;\big|f_\theta^1\wedge\cdots \wedge f_\theta^N\big\rangle\big\langle f_\theta^1\wedge\cdots\wedge f_\theta^N\big|.$$
A computation shows that its one-particle density matrix is exactly $\gamma$, whereas its two-particle density matrix is
\begin{equation}
\Gamma^{(2)}=\mathscr{A} (\gamma\otimes\gamma) \mathscr{A} -\sum_{1\leq j<k\leq K}\left|\sum_{i=1}^NV^i_j\overline{V^i_k}\right|^2|u_j\wedge u_k\rangle\langle u_j\wedge u_k|,
\label{eq:comput_Gamma_2_Lieb}
\end{equation}
where $\mathscr{A}$ is the orthogonal projection onto the anti-symmetric two-particle subspace $L^2(\R^d\times\Z_q,\C)\wedge L^2(\R^d\times\Z_q,\C)$.
The first operator has the integral kernel
\begin{equation}
\mathscr{A} (\gamma\otimes\gamma) \mathscr{A}(\bx_1,\bx_2;\by_1,\by_2)=\gamma(\bx_1,\by_1)\gamma(\bx_2,\by_2)-\gamma(\bx_1,\bx_2)\gamma(\by_1,\by_2).
\label{eq:Gamma2}
\end{equation}
The operator $\mathscr{A} (\gamma\otimes\gamma) \mathscr{A}$ is exactly the two-particle density matrix of the unique quasi-free state over the Fock space that has the one-particle density $\gamma$~\cite{BacLieSol-94}. Integrating against the potential $w$, the first term in~\eqref{eq:Gamma2} gives the classical energy (Hartree term) whereas the second gives the exchange term, which is non-positive for $w\geq0$.

Since the last term in~\eqref{eq:comput_Gamma_2_Lieb} is a non-positive operator, the following was obtained in~\cite{Lieb-81} after using the density of finite rank operators.

\begin{lemma}[Mixed canonical states and quasi-free states~\cite{Lieb-81}]
Let $0\leq\gamma=\gamma^*\leq1$ be a one-particle density matrix such that $\tr(\gamma)=N\in \N$. Then there exists a mixed state $\Gamma$ over the fermionic $N$-particle space $\bigwedge_1^NL^2(\R^d\times\Z_q,\C)$ such that its one-particle density matrix is $\gamma$ and its two-particle density matrix $\Gamma^{(2)}$ satisfies
\begin{equation}
 \Gamma^{(2)}\leq \mathscr{A}(\gamma\otimes\gamma) \mathscr{A}
 \label{eq:compare_Gamma_2}
\end{equation}
in the operator sense.
\end{lemma}

Using the lemma for the trial state~\eqref{eq:trial_state} employed in the proof of Theorem~\ref{thm:upper_bound_T} and neglecting the exchange term, the following was derived in~\cite{LewLieSei-19}.

\begin{theorem}[Upper bound on {$F_{\rm L}[\rho]$~\cite{LewLieSei-19}}]\label{thm:upper_bound_F_L}
For $w(\br)=|\br|^{-s}$ in dimension $d>s>0$, we have
\begin{multline}
F_{\rm L}[\rho]\leq c_{\rm TF}(d)(1+\eps) q^{-\frac2d}\int_{\R^d}\rho(\br)^{1+\frac2d}\,\rd\br+\kappa'(d)\frac{1+\eps}{\eps}\int_{\R^d}|\nabla\sqrt\rho(\br)|^2\,\rd\br\\+\frac12\iint_{\R^d\times\R^d}\frac{\rho(\br)\rho(\br')}{|\br-\br'|^s}\,\rd\br\,\rd\br'
\label{eq:upper_bound_F_L}
\end{multline}
for any $\varepsilon>0$.
\end{theorem}

This time, the right side of~\eqref{eq:upper_bound_F_L} involves an energy functional that looks like the Thomas-Fermi-von Weiz\"acker energy~\cite{Lieb-81b}. It is an open problem to derive a similar upper bound for $F_{\rm LL}$.

\appendix
\section{The Hohenberg-Kohn Theorem} \label{sec:HK}

In this chapter we have mainly discussed the convex formulation of Density Functional Theory~\cite{Lieb-83b,HelTea-22} based on the universal functionals of Levy and Lieb. Another important result is the \emph{Hohenberg-Kohn theorem}~\cite{HohKoh-64} which, in spite of its rather abstract character, is often cited as the main justification for the use of the density to replace the $N$-particle wavefunction. As we will explain, the necessary assumptions for the validity of this theorem are not yet fully understood mathematically. In fact, this result relies on the \emph{unique continuation principle} which is not completely settled for $N$-particle Hamiltonians. Before stating the Hohenberg-Kohn theorem, we therefore start by discussing unique continuation in detail.

For simplicity we assume throughout the whole appendix that there is no spin:
$$q=1.$$
Adding $q$ does not change anything in the following results but it makes the notation a bit heavier.

\subsection{Many-body unique continuation}

We refer for instance to~\cite{Lammert-18} for a discussion on the importance of the unique continuation for the Hohenberg-Kohn theorem. Because unique continuation is a purely local property we allow here external potentials $v$ whose positive part $v_+=\max(v,0)$ is only locally integrable. For simplicity, we assume that its negative part $v_-=\max(-v,0)$ and the interaction potential $w$ are infinitesimally $(-\Delta)$--form-bounded, as was done in the body of the chapter.

\begin{definition}[Many-body unique continuation]\label{def:UCP}
Let
$$v_+\in L^1_{\rm loc}(\R^d,\R),\qquad v_-,w\in L^p(\R^d,\R)+L^\ii(\R^d)$$
with $v_\pm\geq0$ and $p$ satisfying~\eqref{eq:hyp_p}. We say that the potentials $v=v_+-v_-$ and $w$ satisfy the \emph{many-body unique continuation principle} if, for every integer $N\geq1$, (the Friedrichs realization of) $H_N^{v,w}$ satisfies the unique continuation principle in its form domain:
if we have $H_N^{v,w}\Psi=\lambda\Psi$ for some $\lambda\in\R$ and $\Psi\in \cQ(H_N^{v,w})$ with $|\{\Psi=0\}|>0$, then $\Psi\equiv0$.
\end{definition}

The equation $H_N^{v,w}\Psi=0$ is understood in $\cQ(H_N^{v,w})'$, that is,
$$\frac12\int_{\R^{dN}}\nabla \Phi(\bX)^*\cdot\nabla\Psi(\bX)\,\rd\bX+\int_{\R^{dN}}W_N^{v,w}(\bX)\Phi(\bX)^*\Psi(\bX)\,\rd\bX=0$$
for every $\Phi\in \cQ(H_N^{v,w})$ or, equivalently, in the sense of distributions. Recall that the full $N$-body potential is defined by
$$W_N^{v,w}(\br_1,...,\br_N):=\sum_{j=1}^Nv(\br_j)+\sum_{1\leq j< k\leq N}w(\br_j-\br_k).$$

Our formulation of unique continuation is one of the strongest possible, in that it only requires $\Psi$ to vanish on a set of positive measure in order to deduce that $\Psi\equiv0$. This is the property which is needed in the proof of the Hohenberg-Kohn theorem, as we will see. This is sometimes called ``unique continuation on sets of positive measures''. For $v_+\in L^p_{\rm loc}(\R^d,\R^+)$ with $p$ as in~\eqref{eq:hyp_p} it is shown in~\cite{FigGos-92} that any $\Psi\in H^1_{\rm loc}(\R^{dN})$ vanishing on a set of positive measure and solving $H_N^{v,w}\Psi=\lambda\Psi$ must have a point $\bX_0\in \R^{dN}$ where it vanishes to infinite order, that is, such that
$$\forall \alpha>0,\qquad \int_{|\bX-\bX_0|\leq \alpha}|\Psi(\bX)|^2\,\rd\bX=O(r^\alpha).$$
Unique continuation for functions vanishing to infinite order at one point is usually called ``strong unique continuation''. Many authors consider instead the ``weak unique continuation'' problem where $\Psi$ is instead assumed to vanish on an open set, but this is not sufficient for the Hohenberg-Kohn theorem.

\medskip

Unique continuation is a very well studied question. Note first that if we restrict our attention to the potentials for $N$ electrons in a molecule, where
$$v(\br)=-\sum_{m=1}^M\frac{z_m}{|\br-\mathbf{R}_m|},\qquad w(\br-\br')=\frac{1}{|\br-\br'|}$$
then any eigenfunction of $H_N^{v,w}$ is analytic outside of the singularities of the potential~\cite{Morrey-58}, which form a set of zero measure. Therefore it satisfies the unique continuation principle. However, restricting the theory to this very special, though physically relevant, class of potentials is not appropriate in density functional theory. In order to fully understand the density, it is necessary to allow the largest possible class of potentials.

In a famous work~\cite{JerKen-85}, Jerison and Kenig have proved that the (strong) unique continuation principle holds for $-\Delta+W$ in $\R^D$ under the sole assumption that $W\in L^p_{\rm loc}(\R^D)$ with $p$ satisfying~\eqref{eq:hyp_p} and $d$ replaced by $D$. This was then generalized by Koch and Tataru in~\cite{KocTat-06} and many other authors. These results apply to the $N$-particle setting under the condition that
$$W_N^{v,w}\in L_{\rm loc}^{\frac{dN}{2}}(\R^{dN})$$
and this is valid for all $N\geq2$ when
$$v,w\in L_{\rm loc}^{p}(\R^{d})\text{ for all $1\leq p<\ii$.}$$
This is not far from asking that the potentials are locally bounded (in which case the result would follow for instance from the singular Carleman-type estimate proved in~\cite{Regbaoui-97}). We see that $L^p$ conditions are not well adapted to the $N$-particle problem, since they yield $N$-dependent constraints on $v$ and $w$.
More natural assumptions on $v$ and $w$ involve relative bounds with respect to the Laplacian, since such properties are easily propagated to all $N$. For instance if $v$ and $w$ are infinitesimally $(-\Delta)$--form bounded in $\R^d$,
\begin{equation}
\forall\eps>0,\qquad |v|+|w|\leq \eps(-\Delta)+C_{\eps},
\label{eq:form_bounded}
\end{equation}
then so is the $N$-particle potential $W_N^{v,w}$ in $\R^{dN}$ for every $N$. It seems reasonable to conjecture that the many-body unique continuation principle holds under the sole assumption~\eqref{eq:form_bounded}, but this is not known, even for $N=1$. See~\cite{Simon-82} for a similar conjecture in the Kato class, which involves $L^1$ norms instead of $L^2$ norms.

Georgescu \cite{Georgescu-79} and Schechter-Simon \cite{SchSim-80} have provided one of the first results for $N$-body systems with $N$-independent assumptions on $v$ and $w$, but they required the wavefunction to vanish on an open set (weak unique continuation). In a recent work~\cite{Garrigue-18,Garrigue-20}, Garrigue has extended their result to cover the case of functions vanishing on a set of positive measure. His main assumption is that
$$|v_+|^2\1_{B_R}+|v_-|^2+|w|^2\leq \eps (-\Delta)^{\frac32-\delta}+C_{\delta,\eps,R}$$
for some $\delta>0$ and all $\eps,R>0$. This condition in the one-body space $\R^d$ is inherited by $W_N^{v,w}$ in $\R^{dN}$ for every $N\geq1$. After using the Sobolev inequality, the following result was shown in~\cite{Garrigue-20}.

\begin{theorem}[Unique continuation for $L^p$ potentials~{\cite{Garrigue-20}}]\label{thm:UCP_Lp}
Any potentials $v,w$ with $v_+\in L^p_{\rm loc}(\R^d)$ and $v_-,w\in L^p(\R^d)+L^\ii(\R^d)$ with $p>\max(2,2d/3)$
satisfy the many-body unique continuation property of Definition~\ref{def:UCP}.
\end{theorem}

Theorem~\ref{thm:UCP_Lp} now covers Coulomb-type potentials. This is the best result known at the moment for many-body unique continuation. It is an important problem to generalize it to more singular potentials $v$.

\subsection{Main theorem and some open problems}

Let us now state the Hohenberg-Kohn theorem, which says that the density uniquely determines the potential, under the condition that unique continuation holds.

\begin{theorem}[Hohenberg-Kohn]
Let $(v_1)_-,(v_2)_-,w\in L^p(\R^d)+L^\ii(\R^d)$ and $(v_1)_+,(v_2)_+\in L^1_{\rm loc}(\R^d)$ with $p$ as in~\eqref{eq:hyp_p} and assume that $(v_1,w)$ or $(v_2,w)$ satisfies the many-body unique continuation property of Definition~\ref{def:UCP}. If there are two ground states $\Psi_1$ and $\Psi_2$ of, respectively, $H_N^{v_1,w}$ and $H_N^{v_2,w}$ so that $\rho_{\Psi_1}=\rho_{\Psi_2}$, then we have $v_1=v_2+C$ for some constant $C$.
\end{theorem}

The following proof is essentially the one given in~\cite{Lieb-83b,Garrigue-18}.

\begin{proof}
Changing $v_2$ into $v_2-(E_N[v_2]-E_N[v_1])/N$ we can assume that the two ground state energies are equal.
Note that the assumption $\rho_{\Psi_1}=\rho_{\Psi_2}$ implies that $\Psi_1\in \cQ(H^{v_2,w}_N)$ and $\Psi_2\in \cQ(H^{v_1,w}_N)$, the form domains of the two operators defined in~\eqref{eq:form_domain}.
We can  write
\begin{align*}
\pscal{\Psi_1,H_N^{v_1,w}\Psi_1}&= \pscal{\Psi_1,H_N^{v_2,w}\Psi_1}+\int_{\R^d}\rho_{\Psi_1}(\br)\,\Big(v_1(\br)-v_2(\br)\Big)\,\rd\br\\
&\geq \pscal{\Psi_2,H_N^{v_2,w}\Psi_2}+\int_{\R^d}\rho_{\Psi_1}(\br)\,\Big(v_1(\br)-v_2(\br)\Big)\,\rd\br.
\end{align*}
Exchanging the two indices and using that the two densities are equal, we obtain that there is equality everywhere. In particular, $\Psi_1$ is a ground state for $H_N^{v_2,w}$, hence belongs to its operator domain and solves the equation
$$\Big(H_N^{v_1,w}-H_N^{v_2,w}\Big)\Psi_1=\sum_{j=1}^N\big(v_1(\br_j)-v_2(\br_j)\big)\Psi_1=0$$
in the sense of distributions hence also almost everywhere.
Due to the unique continuation principle (for either $v_1$ or $v_2$), we know that $|\{\Psi_1=0\}|=0$ hence this implies that
$$\sum_{j=1}^N\big(v_1(\br_j)-v_2(\br_j)\big)=0$$
for almost every $\br_1,...,\br_N\in\R^d$. Integrating against $f^{\otimes N}$  with $f\in C^\ii_c(\R^d)$, we deduce that
$$N\int_{\R^d}\big(v_1(\br)-v_2(\br)\big)f(\br)\,\rd\br=0.$$
We obtain $v_1=v_2$ a.e., as we wanted.
\end{proof}

\begin{remark}[Mixed states]\rm
There is a similar Hohenberg-Kohn theorem for mixed states. That is, if we have two $N$-particle mixed states $\Gamma_1$ and $\Gamma_2$ supported on the ground state eigenspaces of $H_N^{v_1,w}$ and $H_N^{v_2,w}$ respectively, such that $\rho_{\Gamma_1}=\rho_{\Gamma_2}$, then $v_1=v_2+C$. The proof is similar.
\end{remark}

From Theorem~\ref{thm:UCP_Lp}, we know that the Hohenberg-Kohn theorem holds when $p>\max(2,2d/3)$. We now discuss some consequences of this result.

Let us consider a fixed interaction $w\in L^p(\R^d)+L^\ii(\R^d)$ (in DFT $w$ is usually the Coulomb potential in dimension $d=3$).
We introduce the set of \emph{$v$-representable densities}
\begin{multline*}
\cR_w:=\Big\{\rho_\Psi\ :\ \Psi \text{ ground state of $H_N^{v,w}$ for some $(v,w)$}\\ \text{satisfying the many-body unique continuation in Definition~\ref{def:UCP}}\Big\}.
\end{multline*}
  The Hohenberg-Kohn theorem states that any $\rho\in\cR_w$ arises from a \emph{unique} potential $v$, up to a constant. We remark that the set $\cR_w$ might be quite small. In fact, all the densities $\rho\in\cR_w$ are positive in the following sense. If we had $|\{\rho=0\}|>0$ then the associated $\Psi$ would vanish on a set of infinite measure, which contradicts the unique continuation principle.\footnote{To include densities vanishing on a set of positive measure we have to allow $v_+$ to be infinite on such sets and rephrase the unique continuation principle accordingly.} It is an interesting question~\cite{Garrigue-21,Garrigue-22} to determine how small $\cR_w$ is.

That the set $\cR_w$ might be quite small is not really a problem in density functional theory, since the sought-after ground state density of course always belongs to $\cR_w$. However, the smallness of $\cR_w$ creates some technical difficulties. For instance if we have $\rho\in \cR_w$ with associated potential $v$ then one cannot use the implicit function theorem to determine a potential for other densities in the neighborhood of $\rho$.

In DFT, one important question is to understand the dependence of $\cR_w$  on the interaction potential $w$. In fact, in Kohn-Sham theory the goal is to replace the many-body problem with interaction $w$ by a non-interacting eigenvalue problem. If we have $\rho\in\cR_w\cap \cR_0$ with $\rho=\rho_\Psi$ and $H_N^{v,w}\Psi=0$, then we conclude that there exists a unique potential $v_{\rm KS}$ called the \emph{Kohn-Sham potential} and a normalized ground state $\Psi'$ such that $\rho_{\Psi'}=\rho$ and
$$H^{v+v_{\rm KS},0}\Psi'=\sum_{j=1}^N\left(-\frac{\Delta_{\br_j}}2+v(\br_j)+v_{\rm KS}(\br_j)\right)\Psi'=0.$$
The spectrum of $H^{v+v_{\rm KS},0}$ is determined in terms of the one-particle Kohn-Sham operator $h_{\rm KS}=-{\Delta}/2+v+v_{\rm KS}$. When its eigenvalues satisfy
$$\lambda_N(h_{\rm KS})<\lambda_{N+1}(h_{\rm KS})$$
then $\Psi'$ is unique up to a phase and equal to the Slater determinant
$$\Psi'=(N!)^{-1/2}\det(\phi_i(\bx_j))$$
where the orbitals $\phi_1,...,\phi_N$ are the $N$ first eigenfunctions of $h_{\rm KS}$. Hence the interacting problem has been mapped onto a non-interacting problem.
For this reason, it is desirable that $\cR_w\cap \cR_0$ is as large as possible, perhaps even equal to the whole set $\cR_w$. This question does not seem to have been studied in detail. Understanding the Coulomb case is the main goal of the Kohn-Sham formulation of DFT. See~\cite{Garrigue-22} for numerical results in this direction.

In the original approach of~\cite{HohKoh-64,KohSha-65}, the Hohenberg-Kohn theorem is used to define universal functionals. Namely, for every $\rho\in\cR_w$, we know that there exists a unique potential $v$ (up to constants) and an $N$-particle wavefunction $\Psi$ such that $\rho=\rho_\Psi$ and $\Psi$ is a ground state for $E_N[v]$. Should $\Psi$ be non-degenerate, this defines a map $\rho\mapsto\Psi[\rho]$ and therefore one can define the universal energy functional by $\rho\mapsto \langle\Psi[\rho],H^{0,w}_N\Psi[\rho]\rangle$, and similarly for the kinetic and interaction energies. This approach is not satisfactory from a mathematical point of view, however,  since the set $\cR_w$ is essentially unknown.


\section{Proof of Theorem~\ref{thm:wlsc_F_GC}}\label{app:wlsc_F_GC}

Note that since $w\geq0$, we always have $F_{\rm GC}[\rho]\geq0$. In particular, there is nothing to prove when $\rho=0$. In addition, the Lieb-Thirring inequality~\eqref{eq:Lieb-Thirring} or the Hoffmann-Ostenhof inequality~\eqref{eq:Hoffmann-Ostenhof} imply that $F_{\rm GC}[\rho]>0$ for $\rho\neq0$.

For the proof we need to introduce the $k$-particle density matrices~\cite{Lewin-11}
$$\Gamma^{(k)}=\sum_{n\geq k} \frac{n!}{(n-k)!}\tr_{k+1\to n}[\Gamma_{n}]$$
of a grand-canonical state $\Gamma=(\Gamma_n)_{n\geq0}$, where $\tr_{k+1\to n}$ means the partial trace in the $n-k+1$ last variables. The energy can be expressed in terms of $\Gamma^{(1)}$ and $\Gamma^{(2)}$ only, as follows
$$\sum_{n\geq1}\tr (H^{0,w}_n\Gamma_{n})=\frac12\tr\left((-\Delta)\Gamma^{(1)}\right)+\frac12\tr \left(w_{12}\Gamma^{(2)}\right).$$
Here $w_{12}\geq0$ denotes the multiplication operator by $w(\br_1-\br_2)$ on the two-particle space $L^2(\R^d\times\Z_q,\C)\wedge L^2(\R^d\times\Z_q,\C)$.

We will also use the concept of localized states~\cite{Lewin-11}. For a function $0\leq\chi\leq1$ on $\R^d$, the localized state $\Gamma_{|\chi}$ of a grand-canonical state $\Gamma$ is characterized by the property that its density matrices are equal to $\Gamma_{|\chi}^{(k)}=\chi^{\otimes k}\Gamma^{(k)}\chi^{\otimes k}$ for all $k$.

\subsubsection*{\bf Proof of $(i)$}
Let $\Gamma_j=(\Gamma_{j,n})_{n\geq0}$ be a grand-canonical minimizer for $F_{\rm GC}[\rho_j]$, whose existence  is guaranteed by Theorem~\ref{thm:exists_min_GC}. After extraction of a subsequence we may assume that
$F_{\rm GC}[\rho_j]$ converges to a finite limit (if the limit is $+\ii$ there is nothing to show).

Since the second term is non-negative, the kinetic energy $\tr(-\Delta)\Gamma_j^{(1)}$ must be uniformly bounded, hence $\sqrt{-\Delta}\Gamma_j^{(1)}\sqrt{-\Delta}$ is bounded in the trace-class. For fermions, $\Gamma_j^{(1)}$ is in addition bounded in operator norm by 1.
Note however that we have no \emph{a priori} bound on the number of particles $\tr(\Gamma^{(1)}_j)=\int_{\R^d}\rho_j$, which could diverge. The idea is to use the \emph{local} trace class topology instead.

By the Hoffmann-Ostenhof and Lieb-Thirring inequalities~\eqref{eq:Hoffmann-Ostenhof} and~\eqref{eq:Lieb-Thirring}, $\sqrt{\rho_j}$ is bounded in $\dot{H}^1(\R^d)\cap L^{2+4/d}(\R^d)$. In particular $\rho_j$ is bounded in $L^1$ on any finite ball, uniformly with respect to the center of the ball. This means that $\Gamma_j^{(1)}$ is locally uniformly bounded in the trace-class. Due to the kinetic energy bound, we can therefore assume, after extraction of a subsequence, that $\Gamma_j^{(1)}$ converges strongly locally in the trace class to some operator $\Gamma^{(1)}$ which is such that $\rho_{\Gamma^{(1)}}=\rho$, the weak limit of the sequence $\rho_j$. Since $\rho$ is integrable by assumption, $\Gamma^{(1)}$ is indeed trace-class.

The argument is somewhat more complicated for $\Gamma^{(2)}_j$. Let $0\leq \chi\leq1$ be a function of compact support. Denoting the localized state by $\Gamma_{j|\chi}$, we know that $\Gamma^{(1)}_{j|\chi}=\chi\Gamma^{(1)}_j\chi$  is bounded in the trace class. By Yang's inequality~\cite{Yang-62,Yang-63}, we have $\|\Upsilon_n^{(2)}\|\leq Cn\tr(\Upsilon_n)$ for any $n$-particle operator $\Upsilon_n=\Upsilon_n^*\geq0$ and a universal constant $C$. This implies that
$$\norm{\chi^{\otimes 2}\Gamma_j^{(2)}\chi^{\otimes 2}}=\norm{\Gamma_{j|\chi}^{(2)}}\leq \sum_{n\geq2}\norm{\Gamma_{j|\chi,n}^{(2)}}\leq C\sum_{n\geq2}n\,\tr(\Gamma_{j|\chi,n})\leq C\int_{\R^d}\chi^2\rho_j.$$
Hence $\Gamma_j^{(2)}$ is a locally bounded sequence of operators. No local bound on the trace is known, but this does not create any difficulty. Up to extraction of a subsequence, we can therefore assume that $\Gamma_j^{(2)}\wto \Gamma^{(2)}$ weakly locally as operators.
Since $w_{12}\geq0$ by assumption, Fatou's lemma for operators now implies that
\begin{equation}
\liminf_{j\to\ii}F_{\rm GC}[\rho_j]\geq \frac12\tr\left((-\Delta)\Gamma^{(1)}\right)+\frac12\tr \left(w_{12}\Gamma^{(2)}\right).
\label{eq:wlsc_proof_CAR}
\end{equation}
We have thus shown that the energy is weakly lower semi-continuous when expressed in terms of the one and two-particle density matrices. Our next task is to go back to states in Fock space.

Following~\cite[Lemma~3]{Lewin-11}, we know that there exists a state $\Gamma=(\Gamma_n)_{n\geq0}$ on the Fock space which has the density matrices $\Gamma^{(1)}$ and $\Gamma^{(2)}$. The argument uses a different notion of weak convergence and goes as follows. First we extract a subsequence in the sense of weak-$\ast$ convergence over the local algebra of anti-commutation relations. The weak limit is in principle an abstract state over the local CAR but since its density is $\rho$, which is integrable over $\R^d$, the state is actually normal and arises from a grand-canonical state $\Gamma=(\Gamma_n)_{n\geq0}$~\cite{BraRob1,BraRob2}, which has the above density matrices, as we claimed. In particular, the right side of~\eqref{eq:wlsc_proof_CAR} is by definition $\geq F_{\rm GC}[\rho]$, which concludes the proof of~\eqref{eq:wlsc}.

\subsubsection*{\bf Proof of $(ii)$}
Let $\rho\geq0$ be such that $\sqrt{\rho}\in H^1(\R^d)$ and let $\Gamma=(\Gamma_n)_{n\geq0}$ be a grand-canonical state such that $\rho_\Gamma=\rho$ and
$$F_{\rm GC}[\rho]= \frac12\tr\left(-\Delta\Gamma^{(1)}\right)+\frac12\tr \left(w_{12}\Gamma^{(2)}\right).$$
By optimality, we must have
\begin{equation}
 \tr\big( H^{0,w}_n\Gamma_n\big)=\tr(\Gamma_n)\;F_{\rm L}\left[\frac{\rho_{\Gamma_{n}}}{\tr(\Gamma_{n})}\right]
 \label{eq:optimality_GC_L}
\end{equation}
for all $n\geq1$ such that $\Gamma_n\neq0$ (otherwise we could decrease the energy by choosing another state, without changing the density).
We split the rest of the argument into several steps.

\subsubsection*{Step 1. Approximation by a state with $\rho_{\Gamma_{n_0}}\geq\alpha\rho$}
In this first step we slightly modify $\Gamma$ in order to guarantee that  $\rho_{\Gamma_{n_0}}\geq\alpha\rho$ for some $\alpha>0$ and some $n_0$, a property which will play a role in the next step. To this end we remark that there exists a state $\widetilde\Gamma$ with density $\rho_{\widetilde\Gamma}=\rho$ and which only lives over the $N$ and $(N+1)$--particle subspaces, where $N$ is the integer part of $\int_{\R^d}\rho$. If $\int_{\R^d}\rho\in\N$ we just take $\widetilde\Gamma$ to minimize $F_\text{LL}[\rho]$. Otherwise, we can write $\int_{\R^d}\rho=N+\kappa$ with $\kappa\in(0,1)$ and we consider the state
$$\widetilde\Gamma=(1-\kappa)\Gamma_N+\kappa\Gamma_{N+1}$$
where $\Gamma_N$ optimizes $F_\text{LL}[N\rho/(N+\kappa)]$ and $\Gamma_{N+1}$ optimizes $F_\text{LL}[(N+1)\rho/(N+\kappa)]$. The total density is then
$$\rho_{\widetilde\Gamma}=\left((1-\kappa)\frac{N}{N+\kappa}+\kappa\frac{N+1}{N+\kappa}\right)\rho=\rho,$$
as desired. With the trial state $\tilde\Gamma$ at hand, we consider the new state $\widetilde\Gamma_\eps=(1-\eps)\Gamma+\eps\widetilde\Gamma$. This state has the exact density $\rho$ and its energy converges to that of $\Gamma$ when $\eps\to0$. This new state has the desired property that
$$\rho_{(\widetilde\Gamma_\eps)_N}\geq \eps\rho_{\widetilde\Gamma_N}=\frac{\eps N(1-\kappa)}{N+\kappa}\rho$$
and an energy very close to that of $\Gamma$.

Without loss of generality, we can thus assume for the rest of the proof that we have a state $\Gamma$ with the exact density $\rho_\Gamma=\rho$ and such that $\rho_{\Gamma_{n_0}}\geq \alpha\rho$ for some $n_0$ and $\alpha>0$.

\subsubsection*{Step 2. Approximation by a state with compact support in $n$}
In case that $\Gamma_n$ does not vanish for large $n$, we replace the state $\Gamma=(\Gamma_n)_{n\geq0}$ by a new state $\Gamma'=(\Gamma'_n)$ so that $\Gamma'_n\equiv0$ for $n$ large enough, at the expense of a small error in the energy. Although it is possible to keep the exact density, we will here allow the density to vary a little.

First note that since $\rho\neq0$, we have $F_{\rm GC}[\rho]>0$ and therefore $\Gamma_0<1$ (otherwise the energy would vanish). This allows us to define $\Gamma'$ by
$$\begin{cases}
\Gamma'_0=\Gamma_0+\sum_{n\geq K+1}\tr(\Gamma_n)&\text{for $n=0$,}\\
\Gamma_n'=\Gamma_n & \text{ for $1\leq n\leq K$,}\\
\Gamma_n'=0 & \text{ for $n\geq K+1$.}
\end{cases}$$
In other words, we truncate the state and compensate the missing mass in the vacuum.
The energy of $\Gamma'$ is equal to $\sum_{n=1}^K\tr(H^{0,w}_n\Gamma_n)$ and it converges to $F_{\rm GC}[\rho]$ as $K\to \infty$. Similarly, its density is equal to $\sum_{n=1}^K\rho_{\Gamma_n}$ and it converges to $\rho$ in $L^{1}(\R^d)$. In addition, we have for the one-particle density matrix $\tr(-\Delta)\gamma'\leq \tr(-\Delta)\gamma$ which proves that
\begin{equation}
\int_{\R^d}|\nabla\sqrt{\rho_{\Gamma'}}|^2\leq \sum_{n\geq1}\int_{\R^d}|\nabla\sqrt{\rho_{\Gamma_n}}|^2\leq \tr(-\Delta)\gamma
\label{eq:estim_HO_truncate}
\end{equation}
by the Hoffmann-Ostenhof inequality~\eqref{eq:Hoffmann-Ostenhof_GC}. This gives the strong convergence of $\rho_{\Gamma'}$ to $\rho$ in $L^p(\R^d)$ when $K\to\ii$ for $1\leq p<p^*/2$ where $p^*$ is the critical Sobolev exponent.
Finally, we would like to prove the convergence
\begin{equation}
 \nabla\sqrt{\sum_{n=1}^K\rho_{\Gamma_n}}\to \nabla\sqrt{\sum_{n=1}^\ii\rho_{\Gamma_n}}=\nabla\sqrt\rho,
 \label{eq:CV_gradient}
\end{equation}
strongly in $L^2(\R^d)$ and this is where the first step helps. Indeed, we have the following lemma.

\begin{lemma}\label{lem:CV_gradients}
Let $(\rho_j)$ be a sequence such that $\alpha\rho\leq\rho_j\leq \rho$ for some $\alpha>0$ and $\rho_j(x)\to \rho(x)$ a.e., where $\sqrt\rho\in H^1(\R^d)$. If $\nabla\sqrt{\rho-\rho_j}\to0$ strongly in $L^2(\R^d)$, then $\nabla\sqrt{\rho_j}\to\nabla\sqrt\rho$ strongly in $L^2(\R^d)$.
\end{lemma}
\begin{proof}[Proof of Lemma~\ref{lem:CV_gradients}]
We write $u_j=\sqrt{\rho_j}/\sqrt\rho$, which satisfies $\sqrt\alpha\leq u_j\leq1$ by assumption and which converges almost everywhere to $1$. Then the assumption means that
$$\nabla\sqrt{\rho-\rho_j}=\sqrt{1-u_j^2}\;\nabla\sqrt{\rho}-\sqrt{\rho}\frac{u_j\nabla u_j}{\sqrt{1-u_j^2}}\to0$$
strongly in $L^2(\R^d)$. The first term goes to 0 by dominated convergence, hence we conclude that the second term tends to $0$ in $L^2(\R^d)$. Then we have
\begin{align*}
\int_{\R^d}|\nabla\sqrt\rho-\nabla\sqrt{\rho_j}|^2&=\int_{\R^d}\left|(1-u_j)\nabla\sqrt\rho-\sqrt\rho\nabla u_j\right|^2\\
&\leq 2\int_{\R^d}(1-u_j)^2|\nabla\sqrt\rho|^2+2\int_{\R^d}\rho|\nabla u_j|^2.
\end{align*}
The first term converges again to 0, whereas the second can be estimated by
$$\int_{\R^d}\rho|\nabla u_j|^2\leq \frac{1-\alpha}{\alpha}\int_{\R^d}\frac{\rho u_j^2|\nabla u_j|^2}{1-u_j^2}\to0,$$
and the lemma follows.
\end{proof}

In our case, the inequality~\eqref{eq:estim_HO_truncate} implies that $\sum_{n\geq1}\int_{\R^d}|\nabla\sqrt{\rho_{\Gamma_n}}|^2$ is a convergent series. In addition, by the Cauchy-Schwarz inequality for series,
$$\int_{\R^d}\left|\nabla\sqrt{\sum_{n\geq K+1}\rho_{\Gamma_n}}\right|^2\leq \sum_{n\geq K+1}\int_{\R^d}|\nabla\sqrt{\rho_{\Gamma_n}}|^2\underset{K\to\ii}{\longrightarrow}0.$$
Since we have for $K\geq n_0$
$$\sum_{n=1}^K\rho_{\Gamma_n}\geq \alpha\rho$$
for some $\alpha>0$, the lemma implies the convergence~\eqref{eq:CV_gradient}.

At this step we have replaced $\Gamma=(\Gamma_n)_{n\geq0}$ by a new state $\Gamma'=(\Gamma'_n)$ with compact support in $n$, and a close energy. In addition, $\norm{\rho_{\Gamma'}-\rho_\Gamma}_{L^{1}(\R^d)}$ and  $\norm{\nabla\sqrt{\rho_{\Gamma'}}-\nabla\sqrt{\rho_\Gamma}}_{L^{2}(\R^d)}$ are small. To simplify our exposition we assume henceforth that $\Gamma$ itself satisfies $\Gamma_n\equiv0$ for $n\geq K+1$.

\subsubsection*{Step 3. Approximation by a state with compact support in space}
Next we localize the state $\Gamma$ in order to make it have a compact support in space. Let $\chi_R:=\chi(\cdot/R)$ for some $\chi\in C^\ii_c(\R^d)$ satisfying $\chi(0)=1$ and $0\leq\chi\leq1$. We consider the localized state $\Gamma_{|\chi_R}$ which has the density $\chi_R^2\rho$ and the energy
$$\frac12 \tr((-\Delta) \chi_R\Gamma^{(1)}\chi_R)+\frac12 \tr (w_{12}(\chi^2_R)^{\otimes 2}\Gamma^{(2)}).$$
Note that the space localization does not modify the support in $n$~\cite{Lewin-11}. That is, the state $\Gamma_{|\chi_R}$ satisfies $(\Gamma_{|\chi_R})_n\equiv0$ for $n\geq K+1$ (the same value as for $\Gamma$). The energy and the density converge strongly as $R\to\ii$ to that of $\Gamma$. Hence we can assume in the following that $\Gamma=(\Gamma_n)_{n\geq0}$ has both a compact support in $n$ and in space.

\subsubsection*{Step 4. Construction of the canonical sequence}
In the previous approximations we have replaced the initial grand-canonical state by a new state $\Gamma$ which has an energy close to the minimal energy $F_{\rm GC}[\rho]$ and a density close to the initial density, in all the appropriate function spaces. In this step we finally construct the sequence $\rho_j$ but it will only be close to $\rho_\Gamma$ in the spaces mentioned in the statement of the theorem. We recall that $\Gamma_n\equiv0$ for $n\geq K+1$.

Let $\phi_\ell\in C^\ii_c(B_1)$ be $K$ orthonormal functions in the unit ball $B_1\subset\R^d$ and define $\phi_{j,\ell}(\br)=j^{-d/2}\phi_\ell(\br/j-\bv)$ where $\bv\neq0$ is any fixed vector in $\R^d$. These are  $K$ orthonormal functions in the translated and dilated ball $j\bv +j B_1$. We then introduce the following $K$-particle mixed state
$$\Upsilon_j:=\Gamma_0 |S_{j,0}\rangle\langle S_{j,0}|+\sum_{n=1}^{K-1} \Gamma_n\wedge |S_{j,n}\rangle\langle S_{j,n}|+\Gamma_K$$
where
$$S_{j,n}:=\phi_{j,n+1}\wedge\cdots \wedge\phi_{j,K}$$
is an $(K-n)$-particle Slater determinant.
For $j$ large enough the $\Gamma_n$'s and $\phi_{j,\ell}$ have disjoint support, hence
$$\tr(\Upsilon_j)=\Gamma_0+\sum_{n=1}^N\tr(\Gamma_n)=1,$$
as required. After a lengthy but straightforward calculation, one finds that $\Upsilon_j$ has the energy
\begin{multline*}
\tr(H^{0,w}_{K}\Upsilon_j)=\sum_{n=1}^K\tr(H^{0,w}_n\Gamma_n)+\sum_{n=0}^{K-1}\pscal{S_{j,n}H^{0,w}_{K-n}S_{j,n})}\tr(\Gamma_n)\\
+\sum_{n=0}^{K-1} \int_{\R^d}\int_{\R^d}w(\br-\br')\rho_{\Gamma_n}(\br)\,\rho_{S_{j,n}}(\br')\,\rd\br\,\rd\br'.
\end{multline*}
Under our assumptions on $w$ the last two terms converge to $0$ in the limit $j\to\ii$, hence the energy of $\Upsilon_j$ converges to that of $\Gamma$. In particular,
$$\limsup_{j\to\ii} F_{\rm L}[\rho_{\Upsilon_j}]\leq \lim_{j\to\ii} \tr(H^{0,w}_{K}\Upsilon_j)=\sum_{n=1}^K\tr(H^{0,w}_n\Gamma_n).$$
In addition, the density is
$$\rho_{\Upsilon_j}=\rho_{\Gamma}+\sum_{n=0}^{K-1}\tr(\Gamma_n)\sum_{\ell=n+1}^{K}|\phi_{j,\ell}|^2$$
and its square root converges strongly to $\sqrt{\rho_{\Gamma}}$ in $\dot{H}^1(\R^d)\cap L^{p}(\R^d)$ for all $2<p<p^*$, but not for $p=2$.

\subsubsection*{Step 5. Conclusion}
Using an $\eps/2$ argument to justify the approximations made in Steps 1--3, we have managed to construct the sequence $\rho_j$ mentioned in the statement and proved that it satisfies
\begin{equation}
 \limsup_{j\to\ii}F_{\rm L}[\rho_j]\leq F_{\rm GC}[\rho].
\end{equation}
Next we notice that the lower bound
$$\liminf_{j\to\ii}F_{\rm L}[\rho_j]\geq \liminf_{j\to\ii}F_{\rm GC}[\rho_j] \geq F_{\rm GC}[\rho]$$
follows from $(i)$. Therefore we obtain the stated limit
$$\lim_{j\to\ii}F_{\rm L}[\rho_j]= F_{\rm GC}[\rho]$$
and this completes the proof of the theorem. \qed


\end{document}